\documentclass[12pt,english]{article}
\usepackage{lmodern}
\usepackage{lmodern}
\usepackage[T1]{fontenc}
\usepackage[cp1252]{inputenc}
\usepackage{geometry}
\usepackage{color}
\usepackage{babel}
\usepackage{amsmath}
\usepackage{amsthm}
\usepackage{amssymb}
\usepackage{setspace}
\usepackage[authoryear,comma]{natbib}
\usepackage{microtype}
\usepackage[unicode=true,
 bookmarks=false,
 breaklinks=false,pdfborder={0 0 1},backref=section,colorlinks=true]
 {hyperref}
\hypersetup{pdftitle={"RAUM"},
 pdfauthor={"Dazhuo Wei"},
 pdfnewwindow=true,pdfstartview=FitH,urlcolor=blue!90!red!45!black,citecolor=blue!90!red!45!black,linkcolor=red!90!black}

\makeatletter
\theoremstyle{plain}
\newtheorem{thm}{\protect\theoremname}
\theoremstyle{definition}

\theoremstyle{definition}
\newtheorem{example}[thm]{\protect\examplename}
\theoremstyle{plain}

\usepackage{amsfonts}
\usepackage{dsfont}\usepackage{mathrsfs}\usepackage{ushort}

\usepackage{titlesec}\usepackage{titling}
\usepackage{caption}
\usepackage{enumitem}\usepackage{booktabs}
\usepackage{tikz}
\usepackage{pstricks}\usepackage{pst-all}
\usepackage{pst-plot}
\usepackage{pst-node}\usepackage{pst-3dplot}\usepackage{sgamevar}
\usepackage{subcaption}
\usepackage{lscape}
\usepackage{pifont}
\usepackage{longtable}
\usepackage{algorithmic}
\usepackage{pdflscape}
\usepackage{rotating}
\usepackage{float}

\setenumerate{label=\small(\roman*)}
\DeclareCaptionFont{fancy}{\bfseries\sffamily}
\gamemathtrue
\allowdisplaybreaks

\providecommand{\psreset}{\psset{%
		linewidth=0.3pt,linestyle=solid,linecolor=black,
		dotsize=2.5pt,dotsep=2.5pt,arrowsize=4pt,
		fillstyle=none,fillcolor=white,
		showpoints=false,arrows=-,linearc=0,framearc=0,
		hatchsep=2pt,hatchwidth=0.2pt,nodesep=4pt,opacity=1}
	\psset{gridcolor=black!60, subgridcolor=black!30}
}

\psreset

\usepackage{graphicx}
\graphicspath{ {figures/} }

\titleformat{\section}[block]{\centering\large\bfseries\sffamily}{\thesection.}{0.5em}{}
\titleformat{\subsection}[block]{\flushleft\bfseries}{\thesubsection.}{0.5em}{}
\titleformat{\subsection}[block]{\flushleft\bfseries\sffamily}{\thesubsection.}{0.5em}{}
\titleformat{\subsubsection}[runin]{\normalsize\bfseries\sffamily}{\bfseries\upshape\sffamily\thesubsubsection.}{0.5em}{}[.--\:]
\renewcommand{\thesubsubsection}{\arabic{section}.\arabic{subsection}.\arabic{subsubsection}}
\titlespacing{\section}{0ex}{10ex}{5ex}
\titlespacing{\subsection}{0in}{6ex}{3ex}
\titlespacing{\subsubsection}{0mm}{2ex}{0.5em}
\pretitle{\begin{center}\LARGE\bfseries\sffamily}
\posttitle{\par\end{center}\vskip 0.5em}
\preauthor{\begin{center} \large \lineskip 0.5em\begin{tabular}[t]{c}}
\postauthor{\end{tabular}\par\end{center}}
\predate{\begin{center}\small}
\postdate{\par\end{center}}
\providecommand{\abstitle}[1]{{\par\vspace*{2ex}\small\bfseries\sffamily #1}\hspace*{1ex}}
\renewenvironment{abstract}%
{\begin{center}\begin{minipage}{0.8\linewidth}%
			\abstitle{Abstract}\small}%
		{\end{minipage}\end{center}\vfill\clearpage}

\theoremstyle{remark}
  \theoremstyle{plain}
  \theoremstyle{definition}
    \newtheorem{proposition}{\protect\propositionname}\theoremstyle{definition}
  \newtheorem{definition}{\protect\definitionname}\theoremstyle{plain}
\newtheorem{theorem}{\protect\theoremname}\theoremstyle{plain}
  \theoremstyle{definition}
  \newtheorem{assumption}{\protect\assumptionname}%
  \providecommand{\assumptionname}{Assumption}

  \providecommand{\definitionname}{Definition}
  \providecommand{\lemmaname}{Lemma}
  \providecommand{\propositionname}{Proposition}
  \providecommand{\remarkname}{Remark}
\providecommand{\corollaryname}{Corollary}
\providecommand{\theoremname}{Theorem}
\providecommand{\examplename}{Example}

\makeatother

\providecommand{\definitionname}{Definition}
\providecommand{\examplename}{Example}
\providecommand{\lemmaname}{Lemma}
\providecommand{\theoremname}{Theorem}

\begin{document}
\title{Random Attention Span 
\thanks{I would like to thank Victor Aguiar, Maria Goltsman, Nail Kashaev, Roy Allen, Elliot Lipnowski and John Quah for useful comments and suggestions. I am grateful to Victor Aguiar, Nail Kashaev, Maria Jose Boccardi and Jeongbin Kim for providing their experimental dataset for the empirical estimation chapter in this paper.}} 

\author{ 
	Dazhuo Wei 
	}
\date{This version: April, 2024 / First Version: April, 2022}
\maketitle

\begin{abstract}
In this paper, I introduce a random attention span model (RAS) which uses stopping time to identify decision-makers' behavior under limited attention. Unlike many limited attention models, the RAS identifies preferences using time variation without any need for menu variation. In addition, the RAS allows the consideration set to be correlated with the preference. I also use the revealed preference theory that provides testable implications for observable choice probabilities. Then, I  test the model and estimate the preference distribution using data from M-Turk experiments on choice behaviors that involve lotteries; there is general alignment with the distribution results from logit attention model. 

\end{abstract}

\section{Introduction}

The stochastic revealed preference theory, also known as the random utility model, was introduced by \cite{Mcfadden_1990}. Its aim is to explain how decision-makers behave when choosing alternatives by assuming that they maximize their utility. A major criticism is that the representation assumes individuals consider all alternatives.\footnote{For example, if only two options are available, apples and bananas, and people choose apples more often; then the RUM model with full attention would conclude that apples are preferred over bananas. However, if the model's full consideration assumption is violated, this conclusion may not be accurate as people may not have seen the bananas at all.}\footnote{\cite{RUM_Victor_2021} rejects full consideration with the experimental data consisting of lotteries; \cite{Goeree_2008} shows that limited information about a product contributes to differences in purchasing outcomes; \cite{van_2010} shows that there is considerable variation in choice shares and considerations across laundry detergent brands; \cite{Honka_2017} shows that advertising makes consumers aware of more options, search more, and find better alternatives.} This has motivated the analysis of models that do not assume full consideration. Most of these models use observed menu variation to distinguish between preference and limited attention effects on choices. Examples include models like \cite{RAM_2020} that requires menu variation and relies on monotonicity on the size of menus, and models like \cite{LA_2016}, \cite{Aguair_2017} and \cite{Manzini_Mariotti_2014} which use menu variation and rely on functional-form restrictions on the attention rules. These models have produced exciting results, but obtaining such menu variation in real-world field data can be challenging.\footnote{\cite{RUM_Victor_2021} points out that nobody has a clean menu variation setup for field data.} For one, it is costly for firms to change menus too frequently. Additionally, some major products always stay on the menu, so the preference ranking between these products cannot be obtained due to the lack of variation.

In this paper, I propose a new model called random attention span model (RAS) that can help identify preferences of individuals without any need for menu variation. Additionally, I provide an empirical approach to estimating the preference distribution and to testing the model given a data set of stochastic choices and stopping times. The RAS model can cover a wide range of nonparametric attention rules within a fixed and observed menu. This model has a richer data set that comprises the information on the time it takes individuals to make decisions that enables me to analyze the consideration rules in a nonparametric way and draw implications on the identified preferences for stochastic choices without observing individuals' exact attention. 

The availability of stopping time data sets is a key motivation of RAS. In contrast to the lack of menu variation, most online shopping platforms record stopping time along with individuals' choices. Therefore, a different approach using time variation instead of menu variation would be convenient for investigating preferences under limited attention.

The observation of stopping time is widely used in the neuroscience and psychology literature \citep{DDM_2008, Milo_2010, Krajbich_2012} and has also been introduced to the economics literature by \cite{Reutskaja_2011}. \cite{Reutskaja_2011} designs an experiment to show that limited attention is true among decision-makers by using an eye-tracking device to identify which alternatives each decision-maker sees. The study finds that people are very good at choosing the best alternative among the subset that the decision-maker really pays attention to. This finding supports the limited attention models that indicate decision-makers may have limited attention and are utility maximizing in the limited attention set.

The RAS theory is based on the assumption that individuals' preferences remain stable over time and that their attention span follows a set of rules called "time monotonicity." This set means that as individuals take more time to make decisions, the probability that their attention is concentrated entirely on any subset of the full menu weakly decreases. Essentially, the theory posits that individuals tend to explore more new items as they spend more time on the decision, without forgetting the items they have already seen. For example, if someone is shopping online and starts by looking at item 1, their consideration set is limited to that single item. As time passes, they are likely to see item 2 and other options, but they will not forget about item 1. 


The time monotonicity assumption is valid in sample models with time introduced. In a model of search and satisfaction where the satisfaction thresholds depend on time, a decision-maker will make a choice when: 1. this choice offers an above-threshold utility; 2. this choice is the first item they have ever seen which is above the satisfaction threshold. In this case, if individuals who stop later are more patient and have higher thresholds, the probability that their attention remains in a small range of considerations is smaller. Therefore, as decision-makers who stop later are more patient, the principle of time monotonicity is upheld. In a model of stochastic choice and limited considerations where the probability of any certain item being considered increases over time, the probability of considering anything beyond any given set is non-decreasing over time.\footnote{Please see \ref{sec: examples} for details.}

In this paper, I derive testable implications through linear inequalities in a homogeneous preference setting. I argue that the more time people spend making a decision, the more likely they are to choose the better option. This is because, as they spend more time, they are more likely to see the preferred items. As a result, decision-makers tend to choose the most preferred items with more time, leaving the probability of selecting the less preferred items to be less likely.

When dealing with heterogeneous preferences, the paper shows that the RAS representation is equivalent to having an equality link between attention rules, a preference distribution, and a stochastic choice data set. This equality means that researchers can estimate the preference distribution by minimizing the distance between the data and the model by using a simulation on the attention rule matrix that follows time monotonicity.

In this paper, I tested RAS with experimental data designed and collected by \cite{RUM_Victor_2021} using testing method in \cite{Kitamura_Stoye_2018}. The test does not reject the null hypothesis that the choice data set is generated by RAS. 

I estimate the preference distribution in the experimental data. For easier model comparison, the utility is limited to constant relative risk aversion (CRRA). I compare my estimation with the benchmark model which constrains attention in the logistic function (CRRA-LA).\footnote{A logit attention is as follows: Consider a decision-maker who assigns a positive weight for each non-empty subset of X. Psychologically,$w_A$ is a strength associated with the subset $A$. The probability of considering $A$ in S can be written as $$\mu(A \mid S)=\frac{w_A}{\sum_{A^{\prime} \subset S} w_{A^{\prime}}}.$$}\footnote{A CRRA-LA setup restricts the attention rule to the logit attention and its utility is the CRRA utility.}. The RAS which only uses stopping time variation yields comparable estimation results to CRRA-LA which relies on data on menu variation.  However, the RAS identifies a significantly higher percentage of decision-makers who prefer a more complex lottery. 

The estimation results show that the extra data on stopping time is important in restricting the attention rules and identifying preferences for items which are much less likely to be considered. The RAS uses data with time variation when decision-makers are facing the same menu, while the CRRA-LA \citep{RUM_Victor_2021,LA_2016} estimation uses data when decision-makers are facing all kinds of menus but without considering the stopping time factor.

My approach differs from other works that have used stopping-time observations. These works have focused more on the link between choice accuracy and the time decision-makers spent, conditional on a specific payoff function. In comparison, the RAS uses stopping time observations to derive the preference distribution under limited attention. The RAS also differs from the rational inattention model with time \citep{Rational_inattention_time_2023} and sequential search models \citep{Sequential_1965,Weitzman_1979,Satisficing_aguair} in the assumption of having formed the consideration set. In sequential search models, one fundamental assumption is that people know the payoff distribution. In the rational inattention model, there is some updating process on the beliefs of what the payoff distribution is. However, in the random attention model, decision-makers do not need to know any information about the payoff distribution beyond their consideration set. Therefore, the attention is "random," and the formation of the consideration set may not have an optimal time to stop.

The drifting diffusion model (DDM) is also popular for using stopping time. While both the RAS and the DDM use stopping time, the two models do so from quite different angles. There are fundamental differences in the underlying assumptions of the RAS and the DDM. In the DDM, only two alternatives are present and one alternative is chosen if the net evidence is above a particular threshold. In contrast, in the RAS, multiple alternatives are present; and the decision-maker will choose the most preferred option if it is in their consideration set. The combination of noise and a threshold in the DDM leads to the possibility of mistakes if the threshold is very low or if there is an unusually high level of noise. However, in the RAS, decision-makers are assumed to make the best choice among the alternatives that they actually consider. In the RAS, the assumption of time monotonicity is exogenous to decision-makers' attention and is conditional on time, although the stopping time at the decision is endogenous. In general, in random attention models, researchers rule out the possibility of mistakes, and the choices decision-makers make are based only on their preferences and the consideration set at the time they make the decision. This difference leads the DDM and the RAS to have different focuses of study. The DDM has a focus on choice accuracy, while the RAS has a focus on the preference distribution.

The RAS is close to the random attention model (RAM) \citep{RAM_2020} in terms of the nonparametric assumption on attention rules. The RAM uses a nonparametric restriction named monotonicity in the size of the attention-forming process. It states that the probability of considering a given set cannot increase if the choice set grows. The RAM provides a general framework to test different models of stochastic consideration when preferences are homogeneous. Therefore, their framework is suitable for data on individuals' choices but not so much when the population has a distribution of different preferences. 

\cite{RAUM_2022} proposes a model that involves heterogeneous preferences and random attention. Their approach is very general, which allows for a correlation between attention rules and preference rankings. However, both \cite{RAUM_2022} and \cite{RAM_2020} rely on the use of menu variation. The RAS instead accommodates a nonparametric restriction on attention while proposing a different monotonicity feature (namely time monotonicity). This paper complements \cite{RAUM_2022} and \cite{RAM_2020} in the sense that it generalizes the use of variation in identifying preferences.


The rest of the paper is structured as follows: In Section \ref{sec: model}, I introduce the model setup. In Section \ref{homo}, I present the random attention span model (RAS) in a homogeneous preference setting. In section \ref{hetero}, I introduce the heterogeneous preference version of RAS. Section \ref{methods} has a discussion on the estimation and testing methods. In Section \ref{Empirical Results}, I present the empirical estimation results. In Section \ref{conclusions}, I provide the concluding remarks and a discussion.

\section{Model}\label{sec: model}
In the world of RAS, decision-makers begin by facing an observed fixed menu. Nature then picks a stopping time, consideration set, and preference for each decision-maker. Decision-makers then choose the best alternative according to their preferences within their consideration set.

\usetikzlibrary{arrows.meta,
                chains,
                positioning,
                shapes.geometric
                }
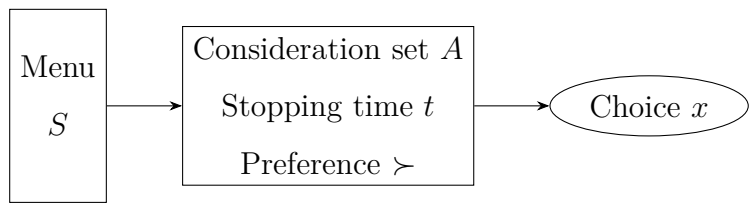
\begin{figure}[h]
\centering
    \begin{tikzpicture}[
    node distance = 10mm and 10mm,
      start chain = going right,
 disc/.style = {shape=cylinder, draw, shape aspect=0.3,
                shape border rotate=90,
                text width=17mm, align=center, font=\linespread{0.8}\selectfont},
  mdl/.style = {shape=ellipse, aspect=2.2, draw},
  alg/.style = {draw, align=center, font=\linespread{1.5}\selectfont}
                    ]
    \begin{scope}[every node/.append style={on chain, join=by -Stealth}]
\node (n1) [alg]  { \\ Menu \\ $S$ \\ };
\node (n2) [alg]  {Consideration set $A$\\ Stopping time $t$ \\Preference $\succ$};
\node (n3) [mdl]  {Choice $x$};
    \end{scope}
    \end{tikzpicture}
\caption{Decision-makers begin by facing an observed fixed menu. Nature then picks a stopping time, consideration set, and preference for each decision-maker. Decision-makers then choose the best alternative according to their preferences within their consideration set.}
\label{Figure 1}
\end{figure}

A finite set $S$ is a universal set of all mutually exclusive alternatives that decision-makers face, also known as the menu. An element of $S$ is denoted as $a$, and its cardinality is $|S|=K$. Let $\mathcal{X}$ denote the set of all non-empty subsets of $S$. In this paper, the menu $S$ is fixed and observed by researchers. Therefore, there is no menu variation being used for revealed preference implications.

The stopping time observation $t$ has a support from $[0,T^*]$ where $T^*$ is the maximum amount of time one can spend. The stopping time here represents the duration of time for the decision-makers to make their choice. Let $\succ$ denote the set containing all possible strict preference rankings to rule out ties.

For a scenario where a decision-maker is browsing through Amazon to buy a TV, the menu $S$ displays all the TV listings on Amazon. Item $a$ refers to a specific TV model listed on Amazon. A stopping time $t$ is the period from the start of browsing the TV section on Amazon to the point where the purchase is completed.

\begin{definition}{(Stochastic Choice Data Set).}
A stochastic choice dataset is a map $\pi: S \times T\rightarrow[0,1]$ such that for all $t\in T$,
\begin{equation*}
\begin{split}
&\pi(a \mid t) \geq 0, \forall a\in S,\\
&\sum_{a \in S} \pi(a \mid t)=1\\
\end{split}
\end{equation*}

\end{definition}

This expression, $\pi(a\mid  t)$, denotes the probability of selecting option $a$ in the menu $S$ at the time at which the decision is made, $t$. This formulation accommodates both deterministic and stochastic decision rules. If $\pi(a\mid t)$ equals either one or zero, then the decision is deterministic. For instance, in an online shopping scenario, $\pi(a\mid  t=5 min)$ is the probability of choosing TV $a$ when customers are shopping on Amazon and making their purchases within exactly five minutes.

It is important to understand that decision-makers may not pay attention to all possible options. To determine the likelihood of agents considering different options, it is critical to define the probabilities of potential consideration sets. For any given menu $S$, every non-empty subset of options may be a potential consideration set, and each set may have a certain probability of being considered. The frequency of each consideration set should fall between zero and one, and the total frequency should add up to one. Formally:

\begin{definition}{(Attention Rule).}
An attention rule is a map $\mu: \mathcal{X} \times T \times \succ \rightarrow[0,1]$ such that for all $A\in \mathcal{X}$, $t\in T$, $\succ_i\in \succ$
\begin{equation*}
\begin{split}
&\mu(A \mid t,\succ_i) \geq 0,\\
&\sum_{A \subseteq S} \mu(A\mid t,\succ_i)=1\\
\end{split}
\end{equation*}

\end{definition}

The notation $\mu(A \mid t,\succ_i)$ represents the probability of focusing explicitly on $A\subseteq S$ with stopping time $t$ when presented with the menu $S$. However, focusing exclusively on $A$ means that the decision-maker is seeing the exact set A, not just some of its subsets. For instance, $\mu(\{a, b\} \mid t,\succ_i)$ represents the probability that the decision-maker is seeing both $a$ and $b$, and not $a$ only or $b$ only. In addition, the possibility of individuals considering nothing or considering anything beyond the menu is ruled out.

Furthermore, by using a conditional attention rule, the RAS is different from \cite{RUM_Victor_2021} in that it does not limit the correlation between attention and a preference. This is especially significant as the preference often determines the consideration sets in various model settings. Conversely, \cite{RUM_Victor_2021} enforces the fact that the random rule for the consideration set and the heterogeneous preferences are independent.

\subsection{Random Attention Span (RAS)}
In this subsection, I introduce the Random Attention Span model with stopping time. Decision-makers of type $i$ have a strict preference $\succ_i$ for X, where $d_{\succ}$ is the number of possible strict preferences. The preference ranking is an asymmetric, transitive, and complete binary relation. With
probability $p(\succ_i)\in [0,1]$, decision-makers are endowed with preferences $\succ_i$ drawn from the set of all linear
orders (strict preferences) on $S$, $\succ$. The probability of all possible preferences sum up to one: $\sum_{\succ_i \in \succ}p(\succ_i)=1$. Decision-makers always pick the best alternative within their consideration set.

\begin{definition}{(Random Attention Span Representation).}
A stochastic choice data set $\pi$ has a random attention span representation if there is a preference ordering set $\succ$ for X, with the distribution of preference $p=[p(\succ_1),p(\succ_2),\cdots,p(\succ_{d_\succ})]$, and an attention rule $\mu$ such that:
$$
\pi(a \mid t)=\sum_{\succ_i \in \succ}p(\succ_i)\sum_{A \subseteq S} 1(a \text { is }-\succ_i \text { best in } A) \cdot \mu(A \mid t, \succ_i)
$$
for all $a\in S$.

\end{definition}
In this case, $\pi$ is represented by $(\mu, p)$. I also define $\pi$ as a random attention span (RAS). This RAS representation is very general and can be used to derive the implications of the revealed preference under time monotonicity.

\subsection{Preference Stability}
The RAS representation is so general that it cannot identify preferences without any further restrictions. The unrestricted distribution of preferences is permissive enough to explain any behavior. The use of random-attention-like models is based on establishing a connection between the data set of stochastic choices and the unobserved attention probabilities and with certain limitations on the attention rules and preference distributions. In the RAS, the attention rules are subject to two fundamental restrictions: time monotonicity and preference stability.

In the RAS, the preference distribution is assumed to be stable over time:
\begin{assumption}{(Preference Stability).}
\label{Assumption 0}
There exists a $p^* \in \Delta(\succ)$ such that $p^*(\succ_i \mid t) = p^*(\succ_i)$ for all $t$, $\succ_i \in \succ$.

\end{assumption}
This stability limitation leads to the marginal distribution of preferences of the general population being equivalent to the true distribution of heterogeneous preferences in the population. The true distribution of preferences is the distribution of choices that controls behavior in a counterfactual situation where there are no limited considerations.

\subsection{Example for RAS Representation}
The RAS representation accommodates \cite{Manzini_Mariotti_2014} while adding attention development with time.

\begin{example}{\cite{Manzini_Mariotti_2014} with Heterogeneous Preference.}

Suppose for a menu $S$, each element $a$ has an probability $\gamma_{t,\succ_i}(a)$ that $a$ is in the consideration set conditional on stopping time $t$ and preference ranking $\succ_i$. Let $\succ$ be the set of all possible preference orderings and $\succ_i \in \succ$. Let $p_t(\succ_i)$ be the probability of $\succ_i$ conditional on time t. Since the marginal distribution of preferences is assumed to be stable over time: $p_t(\succ_i)=p(\succ_i)$ for all $t$ and $\succ_i \in \succ$. The RAS representation is:
\begin{equation*} \label{eq2}
\begin{split}
\pi(a \mid t)=&\sum_{\succ_i \subseteq \succ}p(\succ_i)\sum_{A \subseteq S} 1(a \text { is }-\succ \text { best in } A) \cdot \mu(A \mid t, \succ)\\
=&\sum_{\succ_i \subseteq \succ}[p(\succ_i) \gamma_{t,\succ_i}(a) \prod_{b \in S: b \succ a}(1-\gamma_{t,\succ_i}(b))]
\end{split}
\end{equation*}
Moreover, in the RAS world, $\gamma_{t,\succ_i}(a)$ can depend on the preference ranking, and the correlation between it and the stopping time is expressed by the $\gamma_{t,\succ_i}(a)$. 
\end{example}

In this paper, I explore how people pay attention to different options when making decisions. Imagine there is a menu of all possible choices. At any given time, each option has a chance of catching the decision-maker's attention. When an option does catch their attention, it gets added to the decision-maker's consideration set that is like a mental pocket for options they are considering. As time passes, the collection of options in the consideration set grows. This is the basis for the time monotonicity assumption.

\begin{definition}{(Accumulated Attention)}
\label{Definition 1}

The accumulated attention at A is the probability of paying attention to any subset of A and to A itself:
$$\alpha(A\mid t,\succ_i)=\sum_{B \subseteq A}^{A} \mu(B\mid t,\succ_i) $$
\end{definition}
Logically, the accumulated attention is the probability that one's consideration set is within A or A's subsets and includes nothing apart from the universe of set A.

\begin{assumption}{(Monotonicity in Time).}
\label{Assumption 1}

For any $t<t'$, $A \subset S$,
$$
\alpha(A\mid t,\succ_i) \geqslant \alpha(A\mid t',\succ_i)
$$
and
$$
\alpha(S\mid t,\succ_i) = 1
$$

\end{assumption}

According to Assumption \ref{Assumption 1}, as the stopping time increases, the accumulated attention at any consideration set A should not increase.  In other words, the probability of considering items beyond consideration set A weakly increases for all possible consideration sets. This means that decision-makers are more likely to consider a larger set of alternatives over time. The reason for this consideration is that as time passes, decision-makers tend to move away from smaller consideration sets. The intuition behind this time monotonicity is that the decision-makers are exploring more options as they spend more time.

It is important to understand that with the time monotonicity, for any item $a$ and a given preference, the probability of $a$ being considered conditional on preference is weakly increasing over time. However, the weakly increase in the marginal probability of each item being considered does not guarantee time monotonicity.

\begin{example}

Consider a menu that consists of three items, namely $a$, $b$, and $c$, and two time periods.  

In period 1, $\mu(\{a\}\mid t=1)=0.5$, $\mu(\{b, c\}\mid t=1)=0.5$; 

while in period 2, $\mu(\{a,b\}\mid t=2)=0.5$, $\mu(\{c\}\mid t=2)=0.5$. 

In this example, the probabilities for each item being considered are all weakly increasing, yet the time monotonicity assumption is not satisfied. This is because in period 2, the probability that both b and c are considered together decreases. One reason is that people who considered both b and c in period 1 forget item b in period 2. The time monotonicity assumption rules out this forgetting behavior.

\end{example}

\subsection{Examples for Time Monotonicity}\label{sec: examples}
Many traditional frameworks satisfy time monotonicity. Here are three examples demonstrating why it is a reasonable assumption to make.

\begin{example}{[Top N\citep{Weitzman_1979}.]}

Consider a menu that consists of three items, namely $a$, $b$, and $c$, and three time periods. The search order is fixed. In period 1, only item $a$ is visible to the decision-maker; while in period 2, both $a$ and $b$ are visible; and in period $3$, all three items are visible. Thus, the accumulated attention for all non-empty subsets of the menu $S$ cannot increase. This result can be extended to a menu with $n$ items.

\end{example}

\begin{example}{[\cite{Manzini_Mariotti_2014} Extension.]}

Suppose for a menu $S$ and discrete stopping time $t$ that at each time period $t$, conditional on preference $\succ_i$, each element $a$ has an probability $\gamma_{t,\succ_i}(a)$ that $a$ is in the consideration set conditional on stopping time $t$ and preference ranking $\succ_i$.\footnote{For example, $\gamma_{t=1,a\succ b}(a)$ represents the probability that item $a$ is being considered at time period 1 for people who prefer $a$ to $b$} If I assume people do not forget what they already considered previously, then $\gamma_{t,\succ_i}(a)$ naturally increases as time moves on (the probability that $a$ is considered increases as t increases for all $a\in S$).

Further, the accumulated attention for a certain set $A$ is the probability that one is not considering anything beyond set $A$. Then the accumulated attention can be written as:
$$
\alpha(A\mid t,\succ_i) = \prod_{b \notin A}(1-\gamma_{t,\succ_i}(b))
$$

Since $\gamma_{t,\succ_i}(b)\leq \gamma_{t',\succ_i}(b)$ for all $t<t'$, $b\in S$. Then
$$
\prod_{b \notin A}(1-\gamma_{t,\succ_i}(b)) \geq \prod_{b \notin A}(1-\gamma_{t',\succ_i}(b)), \forall t<t'
$$
Therefore,
$$
\alpha(A\mid t,\succ_i) \geqslant \alpha(A\mid t',\succ_i), \forall t<t'
$$
\end{example}

\begin{example}{[Search Model with Satisficing (Fixed Utility)\citep{Satisficing_aguair}.]}

A search model with satisficing has the following procedure. First, decision-makers will search for items using a random search index in which items having higher search indexes being considered first. Each decision-maker has their own threshold, and once they find the first item that is above this threshold, they will stop searching and choose that item.

For $S=\{y_1,\dots,y_n\}$, let $s=(s_y)_{y\in S}\in \mathbb{R}^n$ be a random search index with a continuous c.d.f. $F_{(s)}$ and 
$\mathbf{u}=\left(\mathbf{u}_{y}\right)_{y \in S} \in \mathbb{R}^n
$ be a fixed utility. The search index represents the order in which an item will be considered; the larger an index an item has, the more up-front this item will be considered. So, for any $y_s$ and $y_z$, $y_s$ is searched before $y_z$ if and only if $s_{y_s}\geq s_{y_z}$. The distribution of search indexes is independent from the preference orderings.

For example, if the fixed utility has the following order: $u_{y_1}>u_{y_2}>u_{y_3}$, then it means that $y_1\succ y_2\succ y_3$.  And if the search index is $s_{y_2}>s_{y_1}>s_{y_3}$, then $y_1$ is searched after $y_2$ but before $y_3$.

$\tau (t)\in \mathbb{R}$ is a random variable where its distribution depends on time $t$. $\tau (t)$ denotes a threshold conditional on time $t$. For any $t'>t$, $\tau(t')$ is assumed to have a first-order stochastic dominance over $\tau(t)$ in which the people who stay longer are drawn from a distribution where they have higher thresholds.

For any consideration set A with cardinality $|A|=N$, denote $y_{{a_1}^A}\succ y_{{a_2}^A}\succ ...\succ y_{{a_n}^A}$, where ${a_i}^A$ represents the ith highest utility element in the consideration set A.

In general, for all $A\not=S$ and $|A|=N$, the accumulated attention can be constructed as:
$$\alpha(A \mid t)=\sum^A_{B\subseteq A}\mu(B \mid t)=\sum_{i=1}^NPr(u_{{a_i}^A}\geq \tau(t))\times P_{{a_i}^A}$$
where $u_{{a_i}^A}$ is the utility index of the $i$th highest utility item in consideration set $A$ and $P_{{a_i}^A}$ is a non-negative number.

And when $A=S$, $$\alpha(S \mid t)=\sum^A_{B\subseteq A}\mu(B \mid t)=1.$$

For example, if there are only three items: $y_1$, $y_2$, and $y_3$; the fixed utility has the following ordering: $u_{y_1}>u_{y_2}>u_{y_3}$, $A=\{y_1, y_2\}$. 
So, I have      
 \begin{equation*} \label{8}
\begin{split}
\alpha(A \mid t)=& \mu(\{y_1, y_2\} \mid t)+\mu(\{y_1\} \mid t)+\mu(\{ y_2\} \mid t)\\
=&[Pr(u_{y_1}\geq \tau(t))-Pr(u_{y_2}\geq \tau(t))]\times Pr(s_{y_{2}}\geq s_{y_{1}}\geq s_{y_{3}})\\
&+Pr(u_{y_{1}}\geq \tau(t))\times Pr(s_{y_{1}}\geq s_{y_{2}}, s_{y_{1}}\geq s_{y_{3}})+Pr(s_{y_{2}}\geq \tau(t))\times Pr(s_{y_{2}}\geq s_{y_{1}}, s_{y_{2}}\geq s_{y_{3}})\\
=&Pr(u_{y_1}\geq \tau(t))\times [Pr(s_{y_{2}}\geq s_{y_{1}}\geq s_{y_{3}})+Pr(s_{y_{1}}\geq s_{y_{2}}, s_{y_{1}}\geq s_{y_{3}})]\\
&+Pr(u_{y_2}\geq \tau(t))\times [Pr(s_{y_{2}}\geq s_{y_{1}}, s_{y_{2}}\geq s_{y_{3}})-Pr(s_{y_{2}}\geq s_{y_{1}}\geq s_{y_{3}})]
\end{split}
\end{equation*}

This equation can be rewritten as:
$$\alpha(A \mid t)= Pr(u_{y_1}\geq \tau(t))\times P_{a_1^A}+Pr(u_{y_2}\geq \tau(t))\times P_{a_2^A}$$
where $P_{a_1^A}= [Pr(s_{y_{2}}\geq s_{y_{1}}\geq s_{y_{3}})+Pr(s_{y_{1}}\geq s_{y_{2}}, s_{y_{1}}\geq s_{y_{3}})]\geq 0$ and $P_{a_2^A}=[Pr(s_{y_{2}}\geq s_{y_{1}}, s_{y_{2}}\geq s_{y_{3}})-Pr(s_{y_{2}}\geq s_{y_{1}}\geq s_{y_{3}})]\geq 0$. 

If $\tau(t')$ has a first-order stochastic dominance over $\tau(t)$ for all $t'>t$, then $Pr(u_{y_{i}}\geq \tau(t)) \geq Pr(u_{y_{i}}\geq \tau(t'))$ for all items in the consideration set A. Then, the accumulated attention probability is non-increasing in t for all strict subsets of the full menu and equals one if $A=S$. Further the monotonicity assumption in time is the following:

For any $t<t'$, $A \subset S$,
$$
\sum_{B \subseteq A}^{A} \mu(B \mid t) \geqslant \sum_{B \subseteq A}^{A} \mu\left(B \mid t'\right)
$$
and
$$
\sum_{B \subseteq S}^{S} \mu(B \mid t) = \sum_{B \subseteq S}^{S} \mu\left(B \mid t'\right)
$$

This is an example where stopping time is endogenous and people who stop at different times have different distributions of satisfactory thresholds. The search model with satisfying behavior satisfies monotonicity in time as long as the conditional threshold $\tau(t')$ has a first-order stochastic dominance over $\tau(t)$ for all $t'>t$.

The intuition of this first-order stochastic dominance is that as stopping time goes up, the distribution of the satisficing threshold moves to the right, meaning that people who stop later are more patient and have higher thresholds. Therefore, as decision-makers who stop later are more patient, the principle of time monotonicity is upheld in \cite{Satisficing_aguair}.

\end{example}

\begin{example}{[Attention Diffusion.]}

This example shows an attention-forming process which mimics \cite{DDM_2008}. However, unlike the DDM in which decision-makers collect information for decisions, decision-makers in this example collect the saliency scores for items to be in their consideration set. And time monotonicity is satisfied if the saliency threshold is non-increasing over time in this example.

The diffusion attention example is as follows: For any item $a_i \in S$, a decision-maker collects its saliency score over time until it reaches a certain saliency threshold $\tau_{a_i}(t)\geq 0$. If the score is above $\tau_{a_i}(t)$ at time $t$, then item $a_i$ is in the consideration set at time $t$.

Evidence accumulation in favor of $a_i$ is represented by the Brownian motion with drift $V_a(t) \equiv v_{a_i} t+B_{a_i,t}$, where $B_{a_i,t}\sim N(0, t\sigma^2)$.

Therefore at time $t$, the probability that item $a_i$ is in the consideration set $A$ is:
\begin{equation*} \label{1}
\begin{split}
\gamma_{t,\succ}(a_i)=& Pr(v_{a_i} t+B_{a_i,t}\geq \tau_{a_i}(t))\\
=& Pr(B_{a_i,t}\geq \tau_{a_i}(t) -v_{a_i} t)\\
=& 1-\Phi(\frac{\tau_{a_i}(t) -v_{a_i} t}{\sqrt{t}\sigma})\\
\end{split}
\end{equation*}

Time monotonicity is satisfied if the saliency threshold does not increase over time.

The derivative of the probability with respect to time is:
\begin{equation*} \label{2}
\begin{split}
\frac{\partial \gamma_{t,\succ}(a_i)}{\partial t}=& -\frac{1}{\sigma}\phi(\frac{\tau_{a_i}(t) -v_{a_i} t}{\sqrt{t}\sigma})\times (\frac{\tau'_{a_i}(t)}{\sqrt{t}}-\frac{\tau_{a_i}(t)}{2\sqrt{t^3}}-\frac{v_{a_i}}{2\sqrt{t}})\\
\end{split}
\end{equation*}

If $\tau'_{a_i}(t)\leq 0$, then $\frac{\partial \gamma_{t,\succ}(a_i)}{\partial t}\geq 0$. So, the probability that item $a_i$ is in the consideration set $A$ increases over time.

Note that the accumulated attention to a certain set $A$ is the probability that one is not considering anything beyond set $A$. Then the accumulated attention can be written as:
$$
\alpha(A\mid t,\succ) = \prod_{b \notin A}(1-\gamma_{t,\succ}(b))
$$

Since $\gamma_{t,\succ}(b)\leq \gamma_{t',\succ}(b)$ for all $t<t'$, $b\in S$. Then
$$
\prod_{b \notin A}(1-\gamma_{t,\succ}(b)) \geq \prod_{b \notin A}(1-\gamma_{t',\succ}(b)), \forall t<t'
$$
Therefore,
$$
\alpha(A\mid t,\succ) \geqslant \alpha(A\mid t',\succ), \forall t<t'
$$
\end{example}

\section{Revealed Preference Implications in a Homogeneous Preference Environment}\label{homo}

In this section, I provide the implications of the revealed preferences solely with stopping time variation when the population share a homogeneous preference. These implications apply for the occasions when firms do not change the menu of goods that they sell but manage to record how long customers take to make choices.

I develop a theorem of the implications of the revealed preferences and a corresponding algorithm to reject the "wrong" preference orderings. To begin with, I first build a lower contour set $A_y$ includes all items that are less preferred than y, $$A_y=\{x\in S \mid y\succ_i x\}$$

\begin{proposition}
\label{lemma 1}
Suppose that the time monotonicity is satisfied. With a lower contour set of y, where $A_y$ includes all items that are less preferred than y,  the RAS representation with preference $\succ_i$ can be characterized by the following condition:
$$
\sum_{i_x \in A_y} \pi\left(i_x\mid  t^{\prime}, \succ_i \right) \leqslant \sum_{i_x \in A_y} \pi(i_x\mid  t, \succ_i) \qquad \forall y\in S, t<t',
$$

\end{proposition}

One implication of Proposition \ref{lemma 1} is that the most preferred item is more likely to be picked over time, while the least preferred item is less likely.

According to Proposition \ref{lemma 1}, if the stochastic data set has RAS representation and satisfies time monotonicity, the probability of selecting the least k favorite items, based on a certain preference ordering, should not increase as individuals spend more time. This reasoning helps to understand the relationship between time and probability in the context of choosing favorite items from a stochastic data set.

Unlike the RAM result \citep{RAM_2020}, Proposition \ref{lemma 1} does not directly reveal the preference between two items, but it provides an algorithm to eliminate incorrect preference orders. The algorithm is as follows:

\begin{definition}{(Rejection Method)}
\label{Rejection Algorithm}
 A preference ordering $\succ=\{a_1\succ a_2\succ a_3\succ ...\succ a_n\}$ is rejected if there exist some i, $t<t'$,  such that the stochastic choice dataset has:
 $$
\sum_{a_x \in A_i} \pi\left(a_x\mid  t^{\prime}, \succ \right) > \sum_{a_x \in A_i} \pi(a_x\mid  t, \succ)
$$
where $A_i=\{a_x:x\geq i\}$ is the lower contour set containing $a_i$ and all the items less favorable than $a_i$.

\end{definition}

The logic of rejection directly follows Proposition \ref{lemma 1}. In a homogeneous preference setup, for a proposed preference ordering
$\succ=\{a_1\succ a_2\succ a_3\succ ...\succ a_n\}$, if there exist a time point $t<t'$ and a number $i<n$, then the probability of picking the last i favorable item (according to the proposed preference ordering) increases, then the proposed preference ordering can not be the true one that allows the stochastic dataset to admit RAS representation. Therefore the proposed preference ordering is rejected.

\begin{theorem}{(Survival Characterization)}
\label{theorem 1}
A set of all possible orders of preferences $\succ$ (which survives from the rejection method with the data) is non-empty if and only if the stochastic choice data set has RAS representation with a homogeneous preference.
\end{theorem}

\begin{proof}{(Theorem \ref{theorem 1})}

    The proof consists of two parts:

    (i) Prove that if the stochastic choice dataset has a RAS representation with a homogeneous preference, then the set of possible preference ordering $\succ$ is non-empty.

    If the stochastic choice data set has a RAS representation with a homogeneous preference $\succ^*$; then according to Proposition \ref{lemma 1}, we know that this is equivalent to the stochastic data set having for all y that $t<t'$,
$$
\sum_{i_x \in A_y} \pi\left(i_x\mid  t^{\prime}, \succ^* \right) \leqslant \sum_{i_x \in A_y} \pi(i_x\mid  t, \succ^*)
$$
where $A_y=\{i_x:i_x\succ^* y\}$ contains all available items that are ranked higher  than item y according to $\succ^*$.

Then the preference  $\succ^*$ survives the rejection algorithm. Therefore, the set of all non-rejected possible orders of preferences is non-empty.

(ii)  Prove that if the set of non-rejected possible preference orderings $\succ$ is non-empty, then the stochastic choice dataset has a RAS representation with a homogeneous preference.

If the set of non-rejected possible preference orderings $\succ$ is non-empty, there exists a preference ordering $\succ^{**}$ such that for all y, $t<t'$, the stochastic dataset satisfies
$$
\sum_{i_x \in A_y} \pi\left(i_x\mid  t^{\prime}, \succ^{**} \right) \leqslant \sum_{i_x \in A_y} \pi(i_x\mid  t, \succ^{**})
$$
where $A_y=\{i_x:i_x\succ^{**} y\}$ contains all available items that are ranked higher than item y according to $\succ^{**}$. 

According to Proposition \ref{lemma 1}, the stochastic dataset has a RAS representation with preference ordering $\succ^{**}$.

\end{proof}

The survival characterization offers a set of all possible preference orderings that allows the stochastic dataset to admit RAS representation if the preference ordering is homogeneous. 

The details of a survival algorithm are as follows:

For n items, suppose the preference ordering in the population is homogeneous:

1. Each single item $i_x$ could be a potential item that is ranked highest. Then define set $C^1_x$ as one with $n-1$ items that excludes the item $i_x$. With the rejection algorithm \ref{Rejection Algorithm}, check if
$$\sum_{i \in C^1_x} \pi\left(i\mid  t^{\prime}\right) \leqslant \sum_{i \in C^1_x} \pi(i\mid  t)$$ for all t and t'.
If not, I reject that $i_x$ is ranked first.

2. Based on the unrejected items $i_x$, check a second item $i_{x2}$ with corresponding $n-2$ item set $C^2_{x2}$ and check if 
$$\sum_{i \in C^2_{x2}} \pi\left(i\mid t^{\prime}\right) \leqslant \sum_{i \in C^2_{x2}} \pi(i\mid  t)$$ for all t and t'.
If not, I reject that $i_x$ is ranked first and $i_{x2}$ is ranked second.

3. Follow the procedure until all the items are checked.

This algorithm allows me to reject the unreasonable preference orderings. Here is one example using this algorithm:
\begin{example}
    Suppose there are three items $a\succ b\succ c$. At time t, let only set $\{b,c\}$ be considered. At time t', one will only consider set $\{a,b,c\}$. So, $\pi(a\mid t')=1$ and $\pi(b\mid t)=1$.
We observe this choice probability yet do not know the preference orderings.

Now following the algorithm:
First check when $i_x=a,b,c$

1.If $i_x=a$,
$\pi(b\mid t')+\pi(c\mid t')=0<\pi(b\mid t)+\pi(c\mid t)=1$, fails to reject that a is ranked first.

If $i_x=b$,
$\pi(a\mid t')+\pi(c\mid t')=1>\pi(a\mid t)+\pi(c\mid t)=0$, rejects that b is ranked first.

If $i_x=c$,
$\pi(b\mid t')+\pi(a\mid t')=1=\pi(b\mid t)+\pi(a\mid t)=1$ fails to reject that c is ranked first. However, since $\pi(c\mid t)=0$ for all t,  c cannot be ranked first. So, c is also rejected.

2. Given a is ranked first, check $i_{x2}=b$ or $c$

If $i_{x2}=b$,
$\pi(c\mid t')=0=\pi(c\mid t)$, fails to reject that a is ranked first and b is ranked second.

If $i_{x2}=c$,
$\pi(b\mid t')=0<\pi(b\mid t)=1$, fails to reject that a is ranked first and c is ranked second. However, since $\pi(c\mid t)=0$ for all t, c as the second preferred item is rejected.

So, we have $a\succ b\succ c$.
\end{example}

Further, this is a set identification of preference orderings. Theorem \ref{theorem 1} does not guarantee point identification.  Further research is needed to guarantee singleton survivor preference sets.

\section{Heterogeneous Preferences}\label{hetero}
However, in a heterogeneous environment, decision-makers' attention can be influenced by their preference type. For instance, in the search model with satisficing, individuals will stop searching when the last item they see offers a payoff above a certain threshold. This means that their attention development will stop when they come across such items. While this feature is reasonable, it makes it difficult for the classic limited attention models to accurately identify the preference distribution.

This section investigates the RAS under heterogeneous preferences. The RAS provides revealed preference implications and potential applications while allowing the attention rule to be correlated with preference type.

The RAS representation can be stated as:
\begin{equation} 
		\label{eq 3}
		U\cdot A \cdot P=\Pi
	\end{equation}
Where $U$ is the matrix consisting of probabilities of different attention sets being considered, conditional on time and preference. $A$ is the choice rule matrix that transforms attention sets into choices given preferences. And $P$ represents the preference distribution.
 
Denote n as the number of alternatives, $d_t$ as the number of clustered time periods, $d_{\succ}$ as the number of all possible orderings, $d_c$ as the number of all possible consideration sets, and denote that $d_u=d_{\succ}\times d_c$.

For example, if the attention rule is as follows:
$$U_{\succ_i}=
\begin{bmatrix}
\mu(a \mid  t_1, \succ_i) & \mu(b \mid  t_1, \succ_i) & ...& \mu(S \mid  t_1, \succ_i)\\
\mu(a \mid  t_2, \succ_i) & \mu(b \mid  t_2, \succ_i) & ...& \mu(S \mid  t_2, \succ_i)\\
... & ... & ...& ...\\
\mu(a \mid  t_{d_t}, \succ_i) & \mu(b \mid  t_{d_t}, \succ_i) & ...& \mu(S \mid  t_{d_t}, \succ_i)\\
\end{bmatrix}_{d_t \times d_c}
$$
and $$U=[U_{\succ_1} \quad U_{\succ_2} \quad ... \quad U_{\succ_{d_{\succ}}}]_{d_t \times d_u}$$

Combining the linear equality restrictions for all menus, choices, consideration sets, and preference orders, I can construct a matrix $A$ that consists of zero and  one that does not depend on $\pi$. This $A$ transforms the attention rule conditional on the preference into a choice probability for the different alternatives conditional on the preference. For example, if an individual has a preference of $a\succ b\succ c$ and the agent's consideration set is $\{b,c\}$, then this agent will choose $b$. The corresponding element in matrix A is one. 

\begin{example}{(Generating A Matrix)}

Suppose there are two items $a$ and $b$ and two possible preferences $a\succ_a b$ and $b \succ_b a$. The A matrix is constructed by blocks of discrete choices conditional on the preference orderings:
$$\footnotesize
A=\begin{bmatrix}
1(a \text { is }-\succ_a \text { best in } \{a\}) & 0 & 1(b \text { is }-\succ_a \text { best in } \{a\})& 0\\
1(a \text { is }-\succ_a \text { best in } \{b\}) & 0 & 1(b \text { is }-\succ_a \text { best in } \{b\})& 0\\
1(a \text { is }-\succ_a \text { best in } \{a,b\}) & 0 & 1(b \text { is }-\succ_a \text { best in } \{a,b\})& 0\\
0& 1(a \text { is }-\succ_b \text { best in } \{a\}) & 0 & 1(b \text { is }-\succ_b \text { best in } \{a\})\\
0&1(a \text { is }-\succ_b \text { best in } \{b\}) & 0 & 1(b \text { is }-\succ_b \text { best in } \{b\})\\
0&1(a \text { is }-\succ_b \text { best in } \{a,b\}) & 0 & 1(b \text { is }-\succ_b \text { best in } \{a,b\})
\end{bmatrix}$$ 

In this example, the first two columns in matrix $A$ contain blocks conditional on the preference orderings that indicate whether one will choose item $a$ under each consideration set.
$$A_a=\begin{bmatrix}
1(a \text { is }-\succ_a \text { best in } \{a\}) & 0 \\
1(a \text { is }-\succ_a \text { best in } \{b\}) & 0 \\
1(a \text { is }-\succ_a \text { best in } \{a,b\}) & 0 \\
0& 1(a \text { is }-\succ_b \text { best in } \{a\}) \\
0&1(a \text { is }-\succ_b \text { best in } \{b\}) \\
0&1(a \text { is }-\succ_b \text { best in } \{a,b\}) 
\end{bmatrix}$$ 

For instance, the block $\begin{bmatrix}
1(a \text { is }-\succ_a \text { best in } \{a\}) \\
1(a \text { is }-\succ_a \text { best in } \{b\}) \\
1(a \text { is }-\succ_a \text { best in } \{a,b\})\\
\end{bmatrix}$ represents whether item $a$ will be chosen conditional on $a\succ_a b$ under different consideration sets. Trivially, $a$ will not be chosen if the consideration set is $\{b\}$, and will be chosen as long as $a$ is the best item (according to the preference ordering) in the consideration set.

Therefore, in this example,
$$A=
\begin{bmatrix}
1 & 0 & 0& 0\\
0 & 0 & 1& 0\\
1 & 0 & 0& 0\\
0 & 1 & 0 & 0\\
0 & 0 & 0 & 1\\
0 & 0 & 0 & 1
\end{bmatrix}
$$

\end{example}

With the rule that at each time point the decision-maker will pick the best among their consideration set, the transformation matrix yields:
$$U\times A=\Pi_{\succ}$$

where the $\Pi_{\succ}$ denotes the probability of choices conditional on preferences:
$$
\Pi_{\succ}=
\begin{bmatrix}
\pi(a \mid  t_1, \succ_1) & \pi(a \mid  t_1, \succ_2) & ...& \pi(a \mid  t_1, \succ_{d_{\succ}})& ...& \pi(i_n \mid  t_1, \succ_{d_{\succ}})\\
\pi(a \mid  t_2, \succ_1) & \pi(a \mid  t_2, \succ_2) & ...& \pi(a \mid  t_2, \succ_{d_{\succ}})&...& \pi(i_n \mid  t_2, \succ_{d_{\succ}})\\
... & ... & ...& ...& ...&...\\
\pi(a \mid  t_{d_t}, \succ_1) & \pi(a \mid  t_{d_t}, \succ_2) & ...& \pi(a \mid  t_{d_t}, \succ_{d_{\succ}})&...& \pi(i_n \mid  t_{d_t}, \succ_{d_{\succ}})\\
\end{bmatrix}_{d_t \times (n\times d_\succ)}
$$

\begin{example}
The setup is the same as the previous example. There are two items $a$ and $b$ and two possible preferences $a\succ_a b$ and $b \succ_b a$. Therefore,
$$A=
\begin{bmatrix}
1 & 0 & 0& 0\\
0 & 0 & 1& 0\\
1 & 0 & 0& 0\\
0 & 1 & 0 & 0\\
0 & 0 & 0 & 1\\
0 & 0 & 0 & 1
\end{bmatrix}
$$

Suppose there are two time periods: $t_1$ and $t_2$. The attention rule is:
$$U=
\begin{bmatrix}
\mu(\{a\} \mid  t_1, \succ_a) & \mu(\{b\} \mid  t_1, \succ_a) & \mu(\{a,b\} \mid  t_1, \succ_a)&...& \mu(\{a,b\} \mid  t_1, \succ_b)\\
\mu(\{a\} \mid  t_2, \succ_a) & \mu(\{b\} \mid  t_2, \succ_a) & \mu(\{a,b\} \mid  t_2, \succ_a)&...& \mu(\{a,b\} \mid  t_2, \succ_b)\\
\end{bmatrix}_{2 \times 6}
$$
Then $$
\Pi_{\succ}=U\times A=
\begin{bmatrix}
\pi(a \mid  t_1, \succ_a) & \pi(a \mid  t_1, \succ_b) & \pi(b \mid  t_1, \succ_a)& \pi(b \mid  t_1, \succ_b)\\
\pi(a \mid  t_2, \succ_a) & \pi(a \mid  t_2, \succ_b) & \pi(b \mid  t_2, \succ_a)& \pi(b \mid  t_2, \succ_b)\\
\end{bmatrix}_{2 \times 2}
$$
where $\pi(a \mid  t, \succ_i)=\sum_{B \subseteq S} 1(a \text { is }-\succ \text { best in } B) \cdot \mu(B \mid t, \succ)$
\end{example}

Let
$p=\begin{bmatrix}
p(\succ_1) & p(\succ_2) & ... & p(\succ_{d_{\succ}})
\end{bmatrix}
$ be the distribution over all possible preferences.
Then $P$ is a block diagonal of $p$ with $n$ (number of items in the menu) blocks. 
$$P=\begin{bmatrix}
p^T & 0 & ... & 0\\
0 & p^T & ... & 0\\
... & ... & ... & ...\\
0 & 0 & ... & p^T\\
\end{bmatrix}_{(n\times d_{\succ)}\times n}
$$

where $$\Pi=
\begin{bmatrix}
\pi(a \mid  t_1) & \pi(b \mid  t_1) & ...& \pi(i_n \mid  t_1)\\
\pi(a \mid  t_2) & \pi(b \mid  t_2) & ...& \pi(i_n \mid  t_2)\\
... & ... & ...& ...\\
\pi(a \mid  t_{d_t}) & \pi(b \mid  t_{d_t}) & ...& \pi(i_n \mid  t_{d_t})\\
\end{bmatrix}_{d_t \times n}$$

\subsection{Theorem with Heterogeneous Preference}
To identify preferences, the characterization of the RAS is needed as follows:

\begin{theorem}
\label{theorem 2}
The following are equivalent:

(i) The stochastic choice dataset admits RAS representation with time monotonicity,

(ii) There exists 

\quad a) A matrix $U \in \mathbb{R}_{+}^{d_t\times d_u}$ such that for all $i\in [1,d_{\succ}]$, $k\in [1,d_t]$, $j\in [1,d_c]$
$$
\sum_{j\in [1,d_c]} U_{k, j+(i-1)\times d_c}=1
$$

 \quad b) $P$ where the sum of the elements in $p'$ is one.
 
  \quad c) $U$ follows time monotonicity inequality (\ref{Assumption 1}).
 
such that the stochastic dataset satisfies the equality:
$$
U\cdot A \cdot P=\Pi
$$

\end{theorem}

In order to find the preference distribution in the stochastic data set, I need a bundle $(U,P)$ with (a) a matrix $U$ that satisfies time monotonicity; (b) a bundle $(U,P)$ that minimizes the distance between $U\cdot A \cdot P$ and the choice data set $\Pi$. This method does not offer point identification since I have more parameters than the choice probabilities, yet I can get a set of preference distributions that are not rejected by theorem \ref{theorem 2} and could potentially be the true distribution under the heterogeneous preference data.

\section{Estimation and Testing Methods}\label{methods}
\subsection{Estimation Method}
One way to estimate the preference distribution is to simultaneously search for preference distribution $\hat p$ and attention rule $\hat U$ together to minimize the distance between stochastic choice data set $\hat \Pi$ and $U\cdot A \cdot P$. However, due to the large size of the restrictions on time monotonicity, this approach is hard to implement.

Therefore, I use simulation. The idea is to simulate a large number of attention rules that satisfy time monotonicity. I use the hit-and-run method mentioned in \cite{Hit_run_2021} to simulate the attention rules (See appendix \ref{hit-run} for details). In each simulation $k$, the preference distribution $\hat P_k$ can be easily found since $U^{sim}_k$ is fixed and time monotonicity is already satisfied. Then, I can find the simulation that offers the smallest minimized distance between $\hat \Pi$ and $U^{sim^*}\cdot A \cdot \hat P^{sim^*}$. The $\hat P^{sim^*}$ is the estimation result for the preference distribution.

With the RAS representation $U\cdot A \cdot P=\Pi$, the estimation method here is to simulate the attention rules $U^{sim}_k$ (according to time monotonicity) and minimize the distance between the estimation and the data by finding the proper $\hat P_k$ given the simulated $U^{sim}_k$. 
\begin{equation} 
		\label{10}
\begin{aligned}
\hat P_k=\operatorname{ArgMin}_{P} & \quad ||  U^{sim}_k \cdot A \cdot P-\hat \Pi ||^2\\
\end{aligned}
	\end{equation}

 and the distance is
\begin{equation} 
		\label{3}
\begin{aligned}
d_k=\operatorname{Min}_{P} & \quad ||  U^{sim}_k \cdot A \cdot P-\hat \Pi ||^2\\
\end{aligned}
	\end{equation}

With K number of simulations, 
\begin{equation} 
		\label{4}
\begin{aligned}
\hat P=\{\hat P_k \mid d_k \geq d_i,  \quad \forall i \leq K\}
\end{aligned}
	\end{equation}
With the choice dataset $\Pi$, one can draw some empirical implications about the preference distribution $P$. In fact, if the conditional choice matrix $U(S)\cdot A$ does not have multicollinearity and the maximum time period is greater than $(n_m-1)!$ ($n_m$ represents the number of items in the menu), then the distribution of heterogeneous preferences can be uniquely identified. 

\subsection{Testing Method}
Theorem 2 allows me to test whether a given stochastic choice rule $\Pi$ can be generated by the RAS representation. Testing for RUM is a well-understood problem
and amounts to solving a quadratic optimization with cone constraints \citep{Kitamura_Stoye_2018}. However, the steps for testing the RAS need to be adjusted since the time monotonicity constraints in the attention rule grow exponentially. Therefore, in order to test whether the data have a RAS representation,  I first simulate a large number N of attention rules $U^{sim}_i$ ($i\in N$) that follows time monotonicity. For each $U^{sim}_i$, I minimize the distance:
\begin{equation} 
		\label{5}
\begin{aligned}
d_i=\operatorname{Min}_{P} & \quad ||  U^{sim}_i \cdot A \cdot P-\Pi ||^2\\
\end{aligned}
	\end{equation}
and the approximated attention rule $\hat U$ is the simulated $U^{sim}_i$ that has the least distance: $\hat U=min\{U^{sim}_i | i\in N \}$. 
In appendix \ref{Simulation}, I have shown that given a large enough number of simulations, the true attention rule lies in the set of $U^{sim}_i$ and should be close enough to $\hat U$. Therefore, I am testing whether the stochastic choice dataset admits RAS with the specific attention rule $\hat U$. According to Theorem 2, this is equivalent to testing the null hypothesis: $$
\inf _{P \in \mathbb{R}_{+}^d}\left\|\Pi-\hat U A P\right\|=0
$$
This testing problem can be directly linked to the testing problem in \cite{Kitamura_Stoye_2018}.

Further, the true choice probability $\Pi$ is not observed, but the realized choice frequencies $\hat\Pi$ are. Given the estimator of the true choice probability $\hat\Pi$, the attention rule estimator $\hat U$ is computed. The test statistics is:
$$
\mathrm{T}_n=n \min _{\left[P-\tau_n \iota / d\right] \in \mathbb{R}_{+}^d}\left(\hat\Pi-\hat U A P\right)^{\prime} \hat{\Omega}^{-}\left(\hat\Pi-\hat U A P\right)
$$
where n is the sample size; $\hat{\Omega}^{-}$ is a generalized inverse of a diagonal
matrix $\hat{\Omega}$ such that the i-th diagonal element $\hat{\Omega}_{i,i}$ is a consistent estimator of the asymptotic variance in the i-th component of $\hat\Pi$; $\tau_n$ is a tuning parameter; and $\iota$ is a vector of ones of dimension d.
Let $\hat\Pi_l$, $l=1, \ldots, L$, be bootstrap replications of $\hat\Pi$. The critical values of $T_n$ are computed following the bootstrap procedure in \cite{Kitamura_Stoye_2018}:

(i) Compute $\hat\eta_{\tau_n}=\hat U A P_{\tau_n}$, where $P_{\tau_n}$ solves

$$
n \min _{\left[P-\tau_n \iota / d\right] \in \mathbb{R}_{+}^d}\left(\hat\Pi-\hat U A P\right)^{\prime} \hat{\Omega}^{-}\left(\hat\Pi-\hat U A P\right)
$$

(ii) Compute $$\hat\Pi_l=\hat\Pi_l-\hat\Pi+\hat\eta_{\tau_n}$$
for every $l=1, \ldots, L$ and $\Omega^*$

(iii) Compute the bootstrap test statistic 
$$
\mathrm{T}^*_{n,l}=n \min _{\left[P-\tau_n \iota / d\right] \in \mathbb{R}_{+}^d}\left(\hat\Pi_l-\hat U A P\right)^{\prime} \hat{\Omega}^{*-}\left(\hat\Pi_l-\hat U A P\right)
$$
$l=1, \ldots, L$.

(iv) Use the empirical distribution of the bootstrap statistic to compute critical values of $T_n$.

For a specific confidence level $\alpha \in (0, \frac{1}{2})$, the decision rule for the test is to reject the null hypothesis of the RAS with a particular $U^{sim}$ value if the test statistic $\mathrm{T}_n$ is greater than the $(1-\alpha)$ quantile of the empirical distribution of the bootstrap statistic that is denoted as $\hat{c}_{1-\alpha}$. It is important to note that if the null hypothesis is incorrect and the asymptotic variance of $\hat\Pi$ is bounded above and is not equal to zero, then the test statistic diverges toward infinity as the sample size increases. Therefore, as the sample size grows, I can reject the wrongly specified null hypothesis with a probability that approaches one.

\section{Empirical Results}\label{Empirical Results}

In this section, the RAS model is tested and applied to estimate the preference distribution in the experimental data collected by \cite{RUM_Victor_2021}. When time variation is used with the RAS and the CRRA preference, the result is consistent with the estimation obtained through the menu variation with a logit-attention rule structure and the CRRA preference \citep{RUM_Victor_2021}.

\subsection{The Experimental Dataset}
The experiment is conducted through Amazon MTurk for a large cross-section with at most two disjoint choice sets per individual. The experimental design produces a stochastic choice data set with variation in the full choice set. The experiment contains six lottery alternatives with different expected values and variances to induce (potential) preference heterogeneity. Table \ref{table:1} shows the alternatives and implied preference orderings if the decision-makers are expected utility maximizers with a CRRA Bernoulli utility function
$$
u(x)= \begin{cases}\frac{x^{1-\sigma}}{1-\sigma}, & \sigma \neq 1 \\ \ln (x), & \sigma=1\end{cases}
$$

\begin{table}[ht]
\caption{Lotteries measured by tokens, expected values, and variance}
		\label{table:1}
\resizebox{\columnwidth}{!}{\begin{tabular}{llllllllll}
\hline
Lottery & Expectation & Variance  & \multicolumn{7}{c}{Preference Rank $u(x)=\frac{x^{1-\sigma}}{1-\sigma}$ With $\sigma$}                           \\ \hline
        & &  & -1  & 0   & 0.25 & 0.3 & 0.5 & 0.75 & 1 \\ \cline{4-10}
(1)  $\frac{1}{2} 50+\frac{1}{2} 0$        & 25.000 & 625.00 & 1   & 1 &  2    &  5   &   5  &  6    &  6 \\
 $(2)$  $\frac{1}{2} 30+\frac{1}{2} 10$       & 20.000 & 100.00 & 5          & 5     &   5   &  2   &  1   &  1    &  1 \\
 $(3)$  $\frac{1}{4} 50+\frac{1}{4} 30+\frac{1}{4} 10+\frac{1}{4} 0$       & 22.500 & 368.75 & 3       & 3  &   4   &   4  &   3  &   4   &  4 \\
 $(4)$  $\frac{1}{4} 50+\frac{1}{5} 48+\frac{3}{20} 14+\frac{2}{5} 0$       & 24.125 & 511.73 & 2     & 2  &   1   &   3  &  4   &  5    &  5 \\
 (5)  $\frac{1}{5} 48+\frac{1}{4} 30+\frac{3}{20} 14+\frac{1}{4} 10+\frac{5}{20} 0$       & 21.625 & 251.11 & 4 & 4   &  3    &  1   & 2    & 3     &  3 \\ 
(o)  12 with probability 1        & 12.000 & 0.00 &  6         &    6      &  6    &   6  &  6   &  2    & 2  \\ \cline{1-10} 
\end{tabular}}
\end{table}

I have restricted the decision-makers to follow CRRA utility function in my empirical estimation. According to \cite{RUM_Victor_2021}, the CRRA-LA model is tested with experimental data and the model is not rejected. Moreover, although I have estimated the preference distribution under the CRRA, the estimation does not require any restrictions on the utility's functional form and is fully nonparametric for utility. The use of the CRRA is only to compare my estimation result with \cite{RUM_Victor_2021}.

The RAS uses an estimation approach that only requires time variation without the need for menu variation. Therefore, I only consider the segment of the data where individuals face the full menu that includes all possible lotteries. This segment of the data consists of 171 observations that comprise the lottery choices made by individuals and the time they took to make their choices.
 
To obtain the choice probability conditional on time $\Pi(t\mid S)$, I cluster the observations into six time periods. The first period is grouped by observations spending zero seconds in making choices. The remaining five time periods are clustered using the K-mean clustering method. Table \ref{table:2} shows the K-mean clustered choice data set. Each row represents the time period the probability is conditional on, and each column represents the probability of picking the according lottery during the time period. Interestingly, all players who made a choice in period 1 (where they spend almost no time in making choices) chose the outside option $l_O$. This choice indicates that at the beginning of the experiment, individuals' consideration set was limited to the outside option.

\begin{table}[ht]
\caption{Kmean Clustered Choice Data}
		\label{table:2}
\resizebox{\columnwidth}{!}{\begin{tabular}{lllllll}
\hline
period & l1          & l2         & l3         & l4          & l5          & lO          \\ \hline
1      & 0           & 0          & 0          & 0   & 0           & 1 \\
2      & 0.142857    & 0.204082   & 0.142857   & 0.173469    & 0.112245    & 0.22449 \\
3      & 0.153061    & 0.326531   & 0.204082   & 0.132653    & 0.102041    & 0.0816327  \\
4      & 0.142857    & 0.459184   & 0.122449   & 0.0510204   & 0.112245    & 0.112245 \\
5      & 0.208333    & 0.395833   & 0.114583   & 0.020833    & 0.20833     & 0.0520083 \\
6      & 0.14        & 0.31       & 0.17       & 0.11        & 0.17        & 0.1           \\ \hline
\end{tabular}}
\end{table}

\subsection{Estimation Results}

The goal here is to find the simulated attention rule $U^{sim}$ and preference distribution $P$ bundle that offers the least distance between the estimation and the data. To reduce the computational burden, I assume the preference follows the CRRA rule with $\sigma \in [-1,1]$. I also assume $l_O$ as the outside option lottery (the robustness check in the appendix shows this assumption is plausible). With this CRRA structural restriction and $l_O$ as the outside option, there are six available preference orderings, and each preference ordering corresponds to a range of risk aversion. Table \ref{table:3} provides a summary of the six preference orderings and their corresponding risk aversion ranges.

\begin{table}[ht]
\caption{CRRA Preference Orderings and Corresponding Risk Aversion}
		\label{table:3}
  \begin{center}

\begin{tabular}{lc} 
Preference Order  & Implied $\sigma$ \\
\hline$l_1 \succ l_4 \succ l_3 \succ l_5 \succ l_2$  & {$[-1,0.2287)$} \\
$l_4 \succ l_1 \succ l_5 \succ l_3 \succ l_2$  & $(0.2287,0.2606)$ \\
$l_4 \succ l_5 \succ l_1 \succ l_3 \succ l_2$  & $(0.2606,0.2728)$ \\
$l_5 \succ l_4 \succ l_2 \succ l_3 \succ l_1$  & $(0.2728,0.2832)$ \\
$l_5 \succ l_2 \succ l_4 \succ l_3 \succ l_1$  & $(0.2832,0.3001)$ \\
$l_2 \succ l_5 \succ l_3 \succ l_4 \succ l_1$  & $(0.3001,1]$
\end{tabular}
\end{center}
\end{table}

Based on my results in Table \ref{table:4}, there is a general alignment with the previous research (Table \ref{table:5}) in terms of preference distribution. This general alignment of results shows that RAS can use time variation instead of menu variation and offer a similar outcome. Importantly, RAS complements random attention models when there is a lack of menu variation in data sets. 

\begin{table}[ht]
\caption{Estimation Results}
    \label{table:4}
  \begin{center}
\begin{tabular}{lcc} 
Preference Order & $\hat{p}$ & Implied $\sigma$ \\
\hline$l_1 \succ l_4 \succ l_3 \succ l_5 \succ l_2$ & $0.30181$ & {$[-1,0.2287)$} \\
$l_4 \succ l_1 \succ l_5 \succ l_3 \succ l_2$ & $0.00000$ & $(0.2287,0.2606)$ \\
$l_4 \succ l_5 \succ l_1 \succ l_3 \succ l_2$ & $0.00000$ & $(0.2606,0.2728)$ \\
$l_5 \succ l_4 \succ l_2 \succ l_3 \succ l_1$ & $0.22979$ & $(0.2728,0.2832)$ \\
$l_5 \succ l_2 \succ l_4 \succ l_3 \succ l_1$ & $0.00000$ & $(0.2832,0.3001)$ \\
$l_2 \succ l_5 \succ l_3 \succ l_4 \succ l_1$ & $0.46839$ & $(0.3001,1]$
\end{tabular}
\end{center}
\end{table}

\begin{table}[ht]
\caption{Estimation Results from CRRA-LA Model with Menu Variation}
		\label{table:5}
  \begin{center}
\begin{tabular}{lcc} 
Preference Order & $\hat{p}$ & Implied $\sigma$ \\
\hline$l_1 \succ l_4 \succ l_3 \succ l_5 \succ l_2$ & $0.30500$ & {$[-1,0.2287)$} \\
$l_4 \succ l_1 \succ l_5 \succ l_3 \succ l_2$ & $0.04905$ & $(0.2287,0.2606)$ \\
$l_4 \succ l_5 \succ l_1 \succ l_3 \succ l_2$ & $0.04905$ & $(0.2606,0.2728)$ \\
$l_5 \succ l_4 \succ l_2 \succ l_3 \succ l_1$ & $0.04905$ & $(0.2728,0.2832)$ \\
$l_5 \succ l_2 \succ l_4 \succ l_3 \succ l_1$ & $0.04905$ & $(0.2832,0.3001)$ \\
$l_2 \succ l_5 \succ l_3 \succ l_4 \succ l_1$ & $0.49880$ & $(0.3001,1]$
\end{tabular}
\end{center}
\end{table}

However, the RAS identifies a higher percentage of decision-makers who prefer lottery 5. The estimation in the CRRA-LA setup, where the full menu variation is necessary, is shown in Table \ref{table:5}. Compared to the CRRA-LA setup, the RAS estimates in Table \ref{table:4} show a similar percentage of decision-makers with preference 1\footnote{Preference 1: $l_1 \succ l_4 \succ l_3 \succ l_5 \succ l_2$} ($30.2\%$ versus $30.5\%$ for preference 1). On the other hand, it shows a slightly lower percentage of decision-makers with preference 6\footnote{Preference 6: $l_2 \succ l_5 \succ l_3 \succ l_4 \succ l_1$} ($46.8\%$ versus $49.9\%$ for lottery 2). Therefore, the estimations for preference 1 and preference 6 are stable in both the RAS and CRRA-LA models.

Further, the RAS identifies a significantly higher percentage of individuals who prefer Lottery 5 the most. Thus, Lottery 5 has the highest number of operations, and one possible reason for this different result is that people tend to avoid considering items that look more complicated, yet they might prefer these items in terms of payoffs and variations. The RAS successfully retrieves this preference under limited attention by assigning a higher weight to the individuals who choose lottery 5 (according to time monotonicity).

The difference in magnitude may come from three sources. First,  the RAS captures the attention-forming process with time monotonicity and provides extra information about the stopping time. This information helps in providing a more accurate restriction on attention and leads to a more precise estimation of the preference distribution. Second, the RAS allows a correlation between the attention and a preference that extracts further information about preferences from the data sets. However, my estimation only uses the part of the full data where individuals face a full menu. Therefore, it explains only that part of the data and not the entire data with full menu variation. It would be interesting to add menu variation to the RAS for a fairer comparison of the two estimations.

The RAS was tested with the experimental data, and the resulting p-value of 0.245 indicates that the null hypothesis cannot be rejected. This nonrejection means that the data are likely generated by the RAS, given the simulated attention rule $U^{sim}$. However, while the null hypothesis is rejected, the p-value is still quite low. This is because the test only provides evidence against one specific attention rule $U^{sim}$, rather than the entire set of possible attention rules. One drawback of this testing method is that rejecting the null hypothesis does not necessarily mean that the entire model is rejected. Instead, it only rejects the model's compatibility with a specific attention rule. But since the null cannot be rejected, the entire model of the RAS cannot be rejected as well.


\section{Conclusions}\label{conclusions}
Menu variation is a 'gold mine' in limited attention models. However, it is crucial for researchers to develop limited attention models when menu variation is not available. In this paper, I develop a random attention span model. The RAS is able to provide preference distribution under limited attention with only time variation instead of menu variation. Moreover, it has the potential to generalize the use of variation in identifying preferences.


The model offers sufficient conditions and an algorithm for rejecting untrue preference orderings. I provide a characterization of the RAS. With the characterization, the RAS is testable with the estimated attention rule. 
I then estimate the preference distribution and test the RAS with experimental data from \cite{RUM_Victor_2021}. When restricted to the constant relative risk aversion (CRRA) preferences, the RAS estimation has a general alignment with the CRRA-LA model \citep{RUM_Victor_2021} that uses menu variation. However, the RAS identifies a significantly higher percentage of decision-makers who prefer a more complex lottery. I also test the RAS model using methods in \cite{Kitamura_Stoye_2018}, and the test fails to reject the null that the choice data set is generated by the RAS.

These results are important and desirable because the RAS uses different pieces of information compared to the CRRA-LA setup \citep{RUM_Victor_2021}.\footnote{RAS uses data with time variation when decision-makers are facing the same menu, while CRRA-LA estimation uses data when decision-makers are facing all kinds of menus but without considering the stopping time factor. }  However, in the experimental data by \cite{RUM_Victor_2021}, 32 menus are assigned to decision-makers; therefore, the number of observations used by the RAS is much smaller than the observations used by the CRRA-LA model. This smaller sample means that the RAS is able to offer estimates of the preference distribution as well as the CRRA-LA with fewer observations. Moreover, menu variations are observed much less than time variation, which gives RAS an advantage in empirical estimation. 

In summary, RAS is a new model for understanding consumer behavior under limited attention when menu variation is not available. This is especially important in digital contexts where attention spans are short and menus don't change often.

\bibliographystyle{apalike}
\phantomsection\addcontentsline{toc}{section}{\refname}\bibliography{ageingbib}

\appendix

\section{Robustness Check for CRRA-LA Estimation}
In section 5, I assume that lottery O is the outside option and is always ranked last. Yet this is not necessarily true. To check the robustness of the estimation in section 5, I withdraw the outside option assumption and put lottery O into regular ranking settings. Table \ref{table:6} show the robustness check result. Compared to Table \ref{table:4}, the results show very similar estimations on preference distribution $\hat{\pi}$; these results indicate that the estimation of $\hat{\pi}$ is pretty robust with the assumption of the outside option. 
The robustness check shows that the estimation in section 6 has very satisfactory robustness.

\begin{table}[]
\caption{Estimation Results Robustness Check}
		\label{table:6}
  \begin{center}
\begin{tabular}{lcc} 
Preference Order & $\hat{p}$ & Implied $\sigma$ \\
\hline$l_1 \succ l_4 \succ l_3 \succ l_5 \succ l_2$ & $0.25744$ & {$[-1,0.2287)$} \\
$l_4 \succ l_1 \succ l_5 \succ l_3 \succ l_2$ & $0.00000$ & $(0.2287,0.2606)$ \\
$l_4 \succ l_5 \succ l_1 \succ l_3 \succ l_2$ & $0.00000$ & $(0.2606,0.2728)$ \\
$l_5 \succ l_4 \succ l_2 \succ l_3 \succ l_1$ & $0.32321$ & $(0.2728,0.2832)$ \\
$l_5 \succ l_2 \succ l_4 \succ l_3 \succ l_1$ & $0.00000$ & $(0.2832,0.3001)$ \\
$l_2 \succ l_5 \succ l_3 \succ l_4 \succ l_1$ & $0.41935$ & $(0.3001,1]$
\end{tabular}
\end{center}
\end{table}

\section{Omitted Details About Example 3}
For $S=\{y_1,y_2,y_3\}$, let $s=(s_y)_{y\in S}\in \mathcal{U}$ be a random search index with a continuous c.d.f. $F_{(s)}$ and 
$\mathbf{u}=\left(\mathbf{u}_{y}\right)_{y \in S} \in \mathcal{U}
$ be a fixed utility. The search index represents how soon an item will be considered, the larger the index that an item has, the sooner this item will be considered. So, for any $y_s$ and $y_z$, $y_s$ appears in the search before $y_z$ if and only if $s_{y_s}\geq s_{y_z}$. The distribution of the search indexes is independent from the preference orderings.

The fixed utility has the following ranking: $u_{y1}>u_{y2}>u_{y3}$, which means that $y_1\succ y_2\succ y_3$. 

$\tau (t)\in R$ denotes a threshold conditional on stopping time t. And for any $t'>t$, $\tau(t')$ has a first-order stochastic dominance over $\tau(t)$. We can view this as the people who stay longer are drawn from a distribution where they have higher thresholds.

Next I am going to show that in this search model with satisfaction, the accumulated attention does not increase with t for all consideration sets as long as the threshold $\tau(t)$ does not decrease during t.

Suppose for a given preference ordering $\succ:y_1\succ y_2\succ y_3$. For any consideration set A with cardinality $|A|=N$, denote $y_{A_1}\succ y_{A_2}\succ ...\succ y_{A_N}$, where $A_i$ represents the ith element in consideration set A.

1. For any $|A|=1$, the accumulated attention conditional on preference ordering $\succ$ in this framework equals to:
$\sum^A_{B\subseteq A}\mu(B \mid t)=\mu(A\mid t)=Pr(u_{y_{A_1}}\geq \tau(t))\times Pr(s_{y_{A_1}}\geq s_y, \forall y\in S-y_{A_1})$

2. For any $|A|=2$, the accumulated attention conditional on preference ordering $\succ$ in this framework equals:
\begin{equation*} \label{6}
\begin{split}
\sum^A_{B\subseteq A}\mu(B\mid t)=&\mu(\{y_{A_1},y_{A_2}\} \mid t)+\mu(\{y_{A_1}\} \mid t)+\mu(\{y_{A_2}\} \mid t)\\
=&Pr(u_{y_{A_1}}\geq \tau(t), u_{y_{A_2}}< \tau(t))\times Pr(s_{y_{A_2}}\geq s_{y_{A_1}}\geq s_y, \forall y\in S-y_{A_1}-s_{y_{A_2}})\\
&+Pr(u_{y_{A_2}}\geq \tau(t))\times Pr(s_{y_{A_2}}\geq s_y, \forall y\in S-y_{A_2})\\
&+Pr(u_{y_{A_1}}\geq \tau(t))\times Pr(s_{y_{A_1}}\geq s_y, \forall y\in S-y_{A_1})\\
=&Pr(u_{y_{A_1}}\geq \tau(t), u_{y_{A_2}}< \tau(t))\times Pr(s_{y_{A_2}}\geq s_{y_{A_1}}\geq s_y, \forall y\in S-y_{A_1}-s_{y_{A_2}})\\
&+Pr(u_{y_{A_2}}\geq \tau(t))\times Pr(s_{y_{A_2}}\geq s_y, \forall y\in S-y_{A_2})\\
&+Pr(u_{y_{A_1}}\geq \tau(t), u_{y_{A_2}}< \tau(t))\times Pr(s_{y_{A_1}}\geq s_y, \forall y\in S-y_{A_1})\\
&+Pr(u_{y_{A_1}}\geq \tau(t), u_{y_{A_2}}\geq \tau(t))\times Pr(s_{y_{A_1}}\geq s_y, \forall y\in S-y_{A_1})
\end{split}
\end{equation*}

Since the preference ordering is fixed, and $A_N$ by definition is the lowest ranked item in consideration set A, then: 
$$Pr(u_{y_{A_{N}}}\geq \tau(t))=Pr(u_{y_{i}}\geq \tau(t), \forall y_i\in A)$$
So
\begin{equation*} \label{7}
\begin{split}
\sum^A_{B\subseteq A}\mu(B\mid t)=&Pr(u_{y_{A_1}}\geq \tau(t), u_{y_{A_2}}< \tau(t))\times Pr(s_{y_{A_2}}\geq s_{y_{A_1}}\geq s_y, \forall y\in S-y_{A_1}-s_{y_{A_2}})\\
&+Pr(u_{y_{A_2}}\geq \tau(t))\times Pr(s_{y_{A_2}}\geq s_y, \forall y\in S-y_{A_2})\\
&+Pr(u_{y_{A_1}}\geq \tau(t), u_{y_{A_2}}< \tau(t))\times Pr(s_{y_{A_1}}\geq s_y, \forall y\in S-y_{A_1})\\
&+Pr(u_{y_{A_1}}\geq \tau(t), u_{y_{A_2}}\geq \tau(t))\times Pr(s_{y_{A_1}}\geq s_y, \forall y\in S-y_{A_1})\\
=&Pr(u_{y_{A_1}}\geq \tau(t), u_{y_{A_2}}< \tau(t))\times Pr(s_{y_{A_2}}\geq s_{y_{A_1}}\geq s_y, \forall y\in S-y_{A_1}-s_{y_{A_2}})\\
&+Pr(u_{y_{A_1}}\geq \tau(t), u_{y_{A_2}}\geq \tau(t))\times Pr(s_{y_{A_2}}\geq s_y, \forall y\in S-y_{A_2})\\
&+Pr(u_{y_{A_1}}\geq \tau(t), u_{y_{A_2}}< \tau(t))\times Pr(s_{y_{A_1}}\geq s_y, \forall y\in S-y_{A_1})\\
&+Pr(u_{y_{A_1}}\geq \tau(t), u_{y_{A_2}}\geq \tau(t))\times Pr(s_{y_{A_1}}\geq s_y, \forall y\in S-y_{A_1})\\
=&Pr(u_{y_{A_1}}\geq \tau(t), u_{y_{A_2}}< \tau(t))\times[Pr(s_{y_{A_2}}\geq s_{y_{A_1}}\geq s_y, \forall y\in S-y_{A_1}-s_{y_{A_2}})\\
&+Pr(s_{y_{A_1}}\geq s_y, \forall y\in S-y_{A_1})]\\
&+Pr(u_{y_{A_1}}\geq \tau(t), u_{y_{A_2}}\geq \tau(t))\times [ Pr(s_{y_{A_2}}\geq s_y, \forall y\in S-y_{A_2})\\
&+Pr(s_{y_{A_1}}\geq s_y, \forall y\in S-y_{A_1})]\\
\end{split}
\end{equation*}

Since $$Pr(u_{y_{A_1}}\geq \tau(t), u_{y_{A_2}}< \tau(t))+Pr(u_{y_{A_1}}\geq \tau(t), u_{y_{A_2}}\geq \tau(t))=Pr(u_{y_{A_1}}\geq \tau(t)),$$

We have      
\begin{equation*} 
\begin{split}
\sum^A_{B\subseteq A}\mu(B \mid t)=&Pr(u_{y_{A_1}}\geq \tau(t))\times [Pr(s_{y_{A_2}}\geq s_{y_{A_1}}\geq s_y, \forall y\in S-y_{A_1}-s_{y_{A_2}})+\\
&Pr(s_{y_{A_1}}\geq s_y, \forall y\in S-y_{A_1})]\\
&+Pr(u_{y_{A_2}}\geq \tau(t))\times [Pr(s_{y_{A_2}}\geq s_y, \forall y\in S-y_{A_2})\\
&-Pr(s_{y_{A_2}}\geq s_{y_{A_1}}\geq s_y, \forall y\in S-y_{A_1}-s_{y_{A_2}})]
\end{split}
\end{equation*}

In general, for all $A\not=S$ and $|A|=N$, the accumulated attention can be constructed as:
$$\sum^A_{B\subseteq A}\mu(B \mid t)=\sum_{i=1}^NPr(u_{y_{A_i}}\geq \tau(t))\times P_{y_{A_i}}$$
where $P_{y_{A_i}}$ represents the probability of some certain search ordering happening specifically correlating to $y_{A_i}$.

And when $A=S$, $$\sum^A_{B\subseteq A}\mu(B \mid t)=1.$$

Since $\tau(t')$ has a first-order stochastic dominance over $\tau(t)$ for all $t'>t$, $Pr(u_{y_{A_i}}\geq \tau(t)) \geq Pr(u_{y_{A_i}}\geq \tau(t'))$ for all items in the consideration set A. Therefore, the accumulated attention probability does not increase with t for all strict subsets of the full menu and equals one if $A=S$. Recall the monotonicity assumption in time:

For any $t<t'$, $A \subset S$,
$$
\sum_{B \subseteq A}^{A} \mu(B \mid t) \geqslant \sum_{B \subseteq A}^{A} \mu\left(B \mid t'\right)
$$
and
$$
\sum_{B \subseteq S}^{S} \mu(B \mid t) = \sum_{B \subseteq S}^{S} \mu\left(B \mid t'\right)
$$

The search model with satisficing behavior satisfies the monotonicity in time as long as the conditional threshold $\tau(t')$ has a first-order stochastic dominance over $\tau(t)$ for all $t'>t$.

\section{Time Monotonicity and Drifting Diffusion Model}
The drifting diffusion model (DDM) is a popular model that uses stopping time. While both the RAS and the DDM use stopping time, the two models have quite different angles from which they use stopping time.

In a classic DDM for value-based decisions, there are two items to choose from, denoted by $\{a,b\}$. Decision-makers must choose one item that they prefer, and their stopping time is noted. The DDM model assumes that decision-makers collect noisy information from both alternatives until they reach some threshold $\lambda>0$ at which they make a choice. The mean information accumulation  is denoted by $v(a)$ and $v(b)$, respectively. If $v(a)>v(b)$, then $a$ is preferred over $b$; otherwise, $b$ is preferred over $a$. Decision-makers are assumed to face instantaneous independent white noise fluctuations as modeled by uncorrelated Wiener processes $W_a$ and $W_b$ during the information accumulation. Evidence accumulation in favor of $a$ and $b$ is then represented by the two Brownian motions with drift $V_a(t) \equiv v(a) t+\sigma W_a(t)$ and $V_b(t) \equiv v(b) t+\sigma W_b(t)$. Thus, the net evidence is:

$$
Z_{a, b}(t) \equiv V_a(t)-V_b(t)=[v(a)-v(b)] t+\sigma \sqrt{2} W(t)
$$

With a threshold $\lambda$, a decision-maker stops evidence accumulation and makes their decision if $Z_{a, b}(t)$ reaches either $\lambda$ or $-\lambda$. The time a decision-maker first reaches either $\lambda$ or $-\lambda$ is called a stopping time in the DDM. 
$$
\mathrm{DT}_{a, b} \equiv \min \left\{t: Z_{a, b}(t)=\lambda \text { or } Z_{a, b}(t)=-\lambda\right\}
$$

The decision-maker will choose $a$ if $Z_{a, b}(t)$ reaches $\lambda$ first, and will choose b otherwise.
$$
\mathrm{DO}_{a, b} \equiv \begin{cases}a & \text { if } Z_{a, b}\left(\mathrm{DT}_{a, b}\right)=\lambda \\ b & \text { if } Z_{a, b}\left(\mathrm{DT}_{a, b}\right)=-\lambda\end{cases}
$$

There are fundamental differences in the underlying assumptions of the RAS and the DDM. The DDM has only two alternatives, and one alternative is chosen if the net evidence (with noisy evidence) is above a particular threshold. In contrast, the RAS has multiple alternatives, and the decision-maker will choose the most preferred option if it is in their consideration set. 

The combination of noise and a threshold in the DDM leads to the possibility of mistakes if the threshold is very low or if there is an unusually high level of noise. However, in the RAS, decision-makers are assumed to make the best choice among the alternatives that they actually consider. In the RAS, the assumption of time monotonicity is exogenous to decision-makers' attention, conditional on time, although the stopping time at the decision is endogenous. 

In general, in random attention models, researchers rule out the possibility of mistakes, and the choices decision-makers make are based only on their preferences and the consideration set at the time they make the decision. This difference leads the DDM and the RAS to have different focuses of study. The DDM focuses on choice accuracy, while RAS focuses on preference distribution.

I incorporate the idea of random attention into the DDM to illustrate the conceptual difference in the following example.
\begin{example}{(DDM with Random Attention).}
Suppose a menu has 2 items $S=\{a,b\}$. The decision-makers must choose one item they prefer, and the stopping time is collected. Decision-makers collect noisy information from both alternatives until they reach some threshold $\lambda>0$ and then make the choice.

Decision-makers have random attention and are following time monotonicity. At each time point, the probability that a decision-maker only sees item $a$ is $u(a|t)$ and the probability that a decision-maker only sees item $b$ is $u(b|t)$. A decision-maker's probability of seeing both $a$ and $b$ is $u(a,b|t)$.

According to time monotonicity, the probability of considering only one item decreases as time passes, and the probability of considering both items increases. That is:
For any $t<t'$,

$$
u(a|t) \geqslant u(a|t')
$$
$$
u(b|t) \geqslant u(b|t')
$$
and 
$$
u(a,b|t) \leqslant u(a,b|t')
$$

If the mean information accumulations are $v(a)$ and $v(b)$ ($a\succ b$ if $v(a)>v(b)$ and $b\succ a$ if $v(b)>v(a)$), then decision-makers are assumed to face instantaneous independent white noise fluctuations modeled by uncorrelated Wiener processes $W_a$ and $W_b$ during the information accumulations. 

Decision-makers cannot start evidence accumulation of an item until that item is in their consideration set. If time is discrete, then $V_a(1)=(1-u(b|t=1))v(a)+\sigma w_a(1)$, where $(1-u(b|t=1))$ is the probability that the decision-maker takes $a$ into consideration so they can accumulate evidence for $a$. $V_a(2)=V_a(1)+(1-u(b|t=2))v(a)+\sigma w_a(2)$ and thus $V_a(t)=v(a) (t-\sum_{0}^{t}u(b|t))+\sigma W_a(t)$.

If time is continuous, the evidence accumulation in favor of a and b is then represented by the two Brownian motions with drift $V_a(t) \equiv v(a) (t-\int_{0}^{t}u(b|t) dt)+\sigma W_a(t)$ and $V_b(t) \equiv v(b) (t-\int_{0}^{t}u(a|t) dt)+\sigma W_b(t)$. Thus, the net evidence is:

$$
Z_{a, b}(t) \equiv V_a(t)-V_b(t)=[v(a)-v(b)] t + v(b) \int_{0}^{t}u(a|t) dt - v(a) \int_{0}^{t}u(b|t) dt +\sigma \sqrt{2} W(t)
$$
 
\end{example}

In the DDM with random attention above, the attention rule enters the model exogenously. Moreover, if the probability of noticing an item is zero, then it is impossible that the net evidence reaches the threshold which is in favor of this item. This impossibility means decision-makers will not pick an item they did not notice at all. And if $t\rightarrow +\infty$, the net evidence converges to the item that the decision-maker prefers. While the DDM with random attention is unable to either predict choices or identify the preference anywhere in between due to the noise and the "satisfactory-like" threshold.

\section{Proof for Proposition \ref{lemma 1}}
\begin{proof}[Proof of Proposition \ref{lemma 1}]

The first step is to prove that (i) implies (ii). For a fixed ranking $\succ=\{i_1\succ i_2\succ i_3\succ ...\succ i_n\}$,
I consider from the lowest ordered item $i_n$, according to the RAS, the probability of choosing the last ordered item should be the probability of only considering that particular item:
$$\pi(i_n\mid t,\succ)=\mu(\{i_n\}\mid t,\succ)$$
$$\pi(i_n\mid t',\succ)=\mu(\{i_n\}\mid t',\succ)$$
For second last ranked item $i_{n-1}$
$$\pi(i_{n-1}\mid t,\succ)=\mu(\{i_{n-1}\}\mid t,\succ)+\mu(\{i_{n-1},i_n\}\mid t,\succ)$$
$$\pi(i_{n-1}\mid t',\succ)=\mu(\{i_{n-1}\}\mid t',\succ)+\mu(\{i_{n-1},i_n\}\mid t',\succ)$$
and so on.

Therefore,
$$\sum_{i_x\in A_y}\pi(i_x\mid t,\succ)=\sum_{B \subseteq A_y}^{A_y} \mu(B \mid  t,\succ)$$
According to monotonicity in time, it directly follows that
$$
\sum_{i_x \in A_y} \pi\left(i_x\mid  t^{\prime},\succ \right) \leqslant \sum_{i_x \in A_y} \pi(i_x\mid  t,\succ)
$$

The second step is to prove that (ii) implies (i). The proof is done by construction. Suppose we see a sequence $(i_1, i_2, i_3, ..., i_n)$ such that for all y, $t<t'$,
$$
\sum_{i_x \in A_y} \pi\left(i_x\mid  t^{\prime}, \succ \right) \leqslant \sum_{i_x \in A_y} \pi(i_x\mid  t, \succ)
$$
First, I claim that this is the ranking sequence $\succ=\{i_1\succ i_2\succ i_3\succ ...\succ i_n\}$. One can construct an attention rule such that for any k and t, $$\mu(A_k\mid  t, \succ)=\pi(a_k\mid  t, \succ)$$
and the probability of paying attention to any attention set other than this lower contour set form $A_i$, is zero.
Then the RAS representation survives as follows:

For any k
\begin{equation*} 
\begin{split}
\pi(a_k\mid t,\succ)&=\sum_{A \subseteq S} 1(a \text { is }-\succ \text { best in } A) \cdot \mu(A \mid t,\succ)\\
&=\mu(A_k\mid  t, \succ)+\sum_{A \subseteq S/A_k} 1(a \text { is }-\succ \text { best in } A) \cdot \mu(A \mid t,\succ)\\
&=\mu(A_k\mid  t, \succ)+0\\
&=\mu(A_k\mid  t, \succ)
\end{split}
\end{equation*}
Therefore if (ii) happens in the stochastic dataset, then there exists an attention rule matrix that satisfies the RAS representation. 
\end{proof}

\section{Proof for Theorem \ref{theorem 2}}
\begin{proof}[Proof of Theorem \ref{theorem 2}]

The direction from (i) to (ii) follows directly from the definition of RAS representation. If the stochastic choice dataset admits RAS representation with time monotonicity, then the true attention rule U should follow the time monotonicity, and each element in U and P should be non-negative and less than one. Which yields $$
U\cdot A \cdot P=\Pi
$$

The direction from (ii) to (i) is as follows. Since $V$ fits the monotonicity inequality restriction. Simply claim the $V$ and $P$ are the true attention rule following time monotonicity and preference distribution. The feature a) ensures V can be a probability measure for attention rule and $p'$ can be the preference distribution. According to equation \ref{eq 3},
$$
V\cdot A \cdot P=\Pi
$$
is a RAS representation.

\end{proof}
\section{Hit-and-Run Algorithm}\label{hit-run}
This paper uses "Hit-and-Run" algorithm \citep{Hit_run_2021} to simulate attention rule matrix $U$, which satisfies time monotonicity. 

For each preference ordering $\succ_i$, first generate a vector of probabilities that decision-makers are paying attention to each consideration set at the beginning $t=0$.
$$
U(S,\succ_i,t=0)=
\begin{bmatrix}
\mu(\{a\} \mid  t=0, \succ_i) & \mu(\{b\} \mid  t=0, \succ_i) &...& \mu(S \mid  t=0, \succ_i) \\
\end{bmatrix}
$$ 

Then the vector of accumulative attention can be calculated
$$
V(S,\succ_i,t=0)=
\begin{bmatrix}
\alpha(\{a\} \mid  t=0, \succ_i) & \alpha(\{b\} \mid  t=0, \succ_i) &...& \alpha(S \mid  t=0, \succ_i) \\
\end{bmatrix}
$$ 
where $\alpha(A\mid t,\succ_i)=\sum_{B \subseteq A}^{A} \mu(B\mid t,\succ_i)$.

In the choice dataset \citep{RUM_Victor_2021}, the outside option lottery o always stays in the menu. At $t=0$,  I only observe decision-makers choosing the outside option. Therefore in the empirical estimation, I assume that the outside option is always considered and only the outside option can be considered at $t=0$. Therefore there are 32 possible consideration sets (given that there are 5 lotteries and 1 outside option). This yields $\mu(\{l_o\} \mid  t=0, \succ_i)$, and therefore 
$$U(S,\succ_i,t=0)=
\begin{bmatrix}
1 & 0 & 0 &... & 0 & 0 \\
\end{bmatrix}_{1 \times 32}$$

According to time monotonicity, for any $t<t'$, $A \subset S$,
$$
\alpha(A\mid t,\succ_i) \geqslant \alpha(A\mid t',\succ_i)
$$
and
$$
\alpha(S\mid t,\succ_i) = 1
$$

This means an attention rule at $t=1$ satisfies time monotonicity only if the accumulative attention vector $V(S,\succ_i,t=1)-V(S,\succ_i,t=0)\leq 0$.

The idea behind the "hit-and-run" algorithm is (i) to pick some initial row $U(S,\succ_i,t=0)$; (ii) to construct a candidate row by moving along a random direction within the constrained set on a randomly chosen distance; (iii) to use a user-specified Monte-Carlo acceptance rule to assign to either the initial point $U(S,\succ_i,t=0)$ or the candidate point $U(S,\succ_i,t=1)$; (iv) to apply steps (ii) and (iii) to $U(S,\succ_i,t=1)$ to construct $U(S,\succ_i,t=2)$; (v) to repeat until the length of the chain reaches user chosen number (maximum of time periods).

To put it short, let $e^t=U(S,\succ_i,t)$. Take some arbitrary $e^t$ that already satisfied time monotonicity. Let $\xi$ be a direction vector. Thus, the candidate vector is
$$e^{t+1}=e^t+\gamma \xi$$
where $\gamma$ determines the scale of the perturbation $\gamma \xi$.

Sign constraints. I start with sign constraints to make sure the attention rule vector is the distribution of attention: $\mu(A_k \mid  t, \succ_i)$ for all $k\in K_c$, where $K_c$ is a set of indexes that correspond to $\mu(A_k \mid  t, \succ_i)$ in $e^t$. Hence, the constraints take the form
$$
\gamma \xi_k \geq-e_k^t, \quad \forall k \in K_c
$$
and
$$
\gamma \xi_k \leq 1-e_k^t, \quad \forall k \in K_c
$$

Define $K_{+}=\left\{k \in K_c: \xi_k>0\right\}$, $K_{-}=\left\{k \in K_c: \xi_k<0\right\}$ and $K_{0}=\left\{k \in K_c: \xi_k=0\right\}$. Then the sign constraints are:
$$
\begin{aligned}
& \frac{1-e_k^t}{\xi_k} \geq \gamma \geq-\frac{e_k^t}{\xi_k}, \quad \forall k \in K_{+} \\
& \frac{1-e_k^t}{\xi_k} \leq \gamma \leq-\frac{e_k^t}{\xi_k}, \quad \forall k \in K_{-}
\end{aligned}
$$
Note that the constraints that correspond to $k\in K_0$ are always satisfied since $1\geq e_k^t\geq 0$. Thus, the sign constraints can be simplified to
$$
\min _{k \in K_{-}}-\frac{e_k^t}{\xi_k} \geq \gamma \geq \max _{k \in K_{+}}-\frac{e_k^t}{\xi_k}
$$
and
$$
\min _{k \in K_{+}}\frac{1-e_k^t}{\xi_k} \geq \gamma \geq \max _{k \in K_{-}}\frac{1-e_k^t}{\xi_k}
$$

Time monotonicity constraints. Denote $$\Psi(t)=\begin{bmatrix}
\psi_1^t & \psi_2^t &...& \psi_K^t \\
\end{bmatrix}$$
where $K=||K_c||$ and $\psi_k^t=\sum_{A_x \subseteq A_k} \xi_k$. This $\psi_k^t$ can be considered as accumulative movement direction for the accumulative attention.

According to time monotonicity, for any $t<t'$, $A \subset S$,
$$
\alpha(A\mid t,\succ_i) \geqslant \alpha(A\mid t',\succ_i)
$$
Denote $\alpha_k^t=\alpha(A_k\mid t,\succ_i)$, the time monotonicity is 
$$
\alpha_k^{t+1}- \alpha_k^{t}\leq 0, \quad \forall k \in K_c
$$
Moreover, 
$$
\alpha_k^{t+1}=\alpha_k^{t}+\gamma \psi_k^t, \quad \forall k \in K_c
$$
Therefore,
$$
\gamma \psi_k^t \leq 0, \quad \forall k \in K_c
$$
Define $K^{\psi}_{+}=\left\{k \in K_c: \psi_k>0\right\}$, $K^{\psi}_{-}=\left\{k \in K_c: \psi_k<0\right\}$ and $K^{\psi}_{0}=\left\{k \in K_c: \psi_k=0\right\}$. Therefore time monotonicity indicates:
$$
\begin{aligned}
& \gamma \geq 0, \quad \forall k \in K^{\psi}_{-} \\
& \gamma \leq 0, \quad \forall k \in K^{\psi}_{+}
\end{aligned}
$$
Note that the time monotonicity indicates $\gamma=0$ unless all elements in the accumulative direction vector $\psi^t$ share the same sign. Therefore to operate the simulation process, the direction vector $\xi^t$ has to satisfy the condition such that all elements in accumulative direction $\psi^t$ have the same sign. Computationally, one can restrict the domain of k, and then the sign of $\gamma$ is automatically defined. Here I restrict that $k \in K^{\psi}_{-}$, and therefore $\gamma \geq 0$.

And the restrictions is as follows:
$$
\begin{aligned}
& \min _{k}-\frac{e_k^t}{\xi_k} \geq \gamma \geq 0, \quad \forall \xi^t \leq 0 \\
& \min _{k}\frac{1-e_k^t}{\xi_k} \geq \gamma \geq 0, \quad \forall \xi^t \geq 0 
\end{aligned}
$$

Therefore, the algorithm for each attention rule conditional on a preference $\succ_i$ generation process is as follows:

(i) Pick some initial row $e^{t}$; 

(ii) Generate a random direction vector $\xi^{t}$, which makes $\psi^{t}$ contain only elements that share the same sign or 0.

(iii) Given the direction vector $\xi^{t}$, the range for the scale of the perturbation $\gamma$ can be determined:
$$
\begin{aligned}
& \min _{k}-\frac{e_k^t}{\xi_k} \geq \gamma \geq 0, \quad \forall \xi^t \leq 0 \\
& \min _{k}\frac{1-e_k^t}{\xi_k} \geq \gamma \geq 0, \quad \forall \xi^t \geq 0 
\end{aligned}
$$

(iv) Randomly pick $\gamma$ within the range, and get the next row $e^{t+1}=e^t+\gamma \xi^t$

(v) Repeat until the length of the chain reaches the number of maximum time periods.

(vi) The attention rule conditional on a preference $\succ_i$ is generated. Denoted as $U(\succ_i)$

The attention rule under heterogeneous preference is simply the stack of all simulated conditional attention rules:
$$U^{sim}=[U(\succ_1), U(\succ_2), ... ,U(\succ_n)]$$

\section{A Simulation Exercise}\label{Simulation}
To show that simulation can cover the true attention rule and preference distribution if the number of simulations is large enough. I take the approach of simulation on attention rules. Specifically, I simulate a matrix of attention rules $U^{nsim}$ which satisfies time monotonicity. For example, in three items and two time periods case, the simulated attention matrix is as follows:
$$
U^{nsim}=
\begin{bmatrix}
\mu^{nsim}(a \mid  t_1, \succ_1) & \mu^{nsim}(b \mid  t_1, \succ_1) & ...& \mu^{nsim}(a,b,c \mid  t_1, \succ_6)\\
\mu^{nsim}(a \mid  t_2, \succ_1) & \mu^{nsim}(b \mid  t_2, \succ_1) & ...& \mu^{nsim}(a,b,c \mid  t_2, \succ_6)\\
\end{bmatrix}
$$
Each roll contains all attention rules conditional on 6 preference orderings (since there are 3 items) and different time periods ($t_1$ for roll 1 and $t_2$ for roll 2 in this case).

Denote the observed stochastic dataset as:
$$\Pi(S)=\begin{bmatrix}
\Pi(t_1\mid S)\\
\Pi(t_2\mid S)\\
\end{bmatrix}$$ 

Therefore the RAS representation can be stated as follows:
$$\Pi(S)=U*A*P^{diag}$$
where A is the projection matrix transforming attention rule U to choice probability conditional on preference, and $P^{diag}$ is the block diagonal matrix of preference distribution $P$.

This section performs a test showing with a large enough simulation on the attention rule, the true parameter will be captured.

There are 3 time periods and 3 items. And in each period, the probabilities of choosing each item are recorded.I set $P^*=\begin{bmatrix}
1/6 & 1/6 & 1/6 & 1/6 & 1/6 & 1/6)
\end{bmatrix}
$.

I randomly pick a truth attention rule $U^*$ and have the choice dataset $\Pi^*$ and they stay constant across all the simulation procedures.

The testing consists of two layers of simulation, I call it the inner layer (with N times simulation) and the outer layer (with T times of simulation).

For each outer layer $k=1,2,...,T$, I simulate N times of attention rules $U_{k,i}$, where $i=1,2,...,N$.  Define $d(U_{k,i}*A*P^{diag}, \Pi(S))$ as the Euclidean distance between the simulated $U_{k,i}*A*P^{diag}$ and the stochastic dataset $\Pi(S)$.

Define $d_k$ as the shortest distance between the simulation and the dataset across the inner layer simulations: $$d_k=Inf_i d(U_{k,i}*A*P^{diag}, \Pi(S)) $$ We check if the distance between the truth and simulated estimates are close enough $d_k \leq \delta$, here I pick $\delta=0.05$. If yes, then denote $c_k=1$, else $c_k=0$.

The estimate of probability that $P$ lies in the range of estimated distribution $\hat{P_k}$ is $\eta=\frac{\sum^T_{k=1}c_k}{T}$.

I pick $T=10000$, $N=10, 100, 1000, 10000$. Here I denote the probability that $P$ lies in the range of estimated distribution $\hat{P_k}$ when N=10 as $\eta_{10}$. 

The following table is the result:

$$
\begin{array}{cccc}
\eta_{10} & \eta_{100} & \eta_{1000} & \eta_{10000} \\
0.4528 & 0.7963 & 0.9021 & 0.9891
\end{array}
$$

Theoretically, and N goes to infinity, $\eta$ should converge to 1. Meaning that if we have a large enough simulation of attention rules, it will cover the true attention distribution.

This test also shows that with a large number of simulated attention rules, one can capture a range of preference distributions by simply solving the linear equation based on each simulated attention rule. 
\end{document}